\documentclass[runningheads]{llncs}
\usepackage[T1]{fontenc}
\usepackage{amsmath}
\usepackage{bbding}
\usepackage{graphicx}
\usepackage{listings,amssymb,mathrsfs,hyperref,stmaryrd}
\usepackage{lmodern}
\usepackage[normalem]{ulem}
\usepackage{semantic}
\usepackage{tikz}
\usetikzlibrary{arrows}
\usepackage{algorithm, algorithmic}
\lstdefinelanguage{polyc}{
  morekeywords={for, if, else, return, istring, iint, int, bool, size, true, false},
  sensitive=true,
  comment=[l]{//},
  morecomment=[s]{/*}{*/},
  morestring=[b]',
  escapeinside=``,
  morestring=[b]"
}
\usepackage{xcolor}
\usepackage{color}

\begin{document}
\title{A Programming Language for Feasible Solutions}
%
%
\author{Weijun Chen \and
Yuxi Fu{\Envelope} \and Huan Long
}
\authorrunning{W. Chen et al.}
%
\institute{BASICS, Shanghai Jiao Tong University, Shanghai, China.\\
\email{cwj2018@sjtu.edu.cn},
\email{fu-yx@cs.sjtu.edu.cn},
\email{longhuan@sjtu.edu.cn}
}
\maketitle              
\begin{abstract}
Runtime efficiency and termination are crucial properties in the studies of program verification.
Instead of dealing with these issues in an ad hoc manner, it would be useful to develop a robust framework in which such properties are guaranteed by design.
This paper introduces a new imperative programming language whose design is grounded in a static type system that ensures the following equivalence property:
All definable programs are guaranteed to run in polynomial time; Conversely, all problems solvable in polynomial time can be solved by some programs of the language.
The contribution of this work is twofold.
On the theoretical side, the foundational equivalence property is established, and the proof of the equivalence theorem is non-trivial.
On the practical side, a programming approach is proposed that can streamline program analysis and verification for feasible computations.
An interpreter for the language has been implemented, demonstrating the feasibility of the approach in practice.

\end{abstract}
%
%
%
%
\def\classP{\mathbf{P}}
\def\classFP{\mathbf{FP}}
\def\classPC{\mathbf{FPC}}
\def\classL{\mathscr{L}}
\def\classE{\mathscr{E}}
\def\classB{\mathbf B}
\def\classPH{\mathbf{PH}}
\def\classBFF{\mathbf{BFF}_2}
\def\N{\mathbb{N}}
\def\size#1{\vert#1\vert}
\def\typeInt{\mathtt{int}}
\def\typeBool{\mathtt{bool}}
\def\typeIterInt{\mathtt{iint}}
\def\opSize{\mathtt{size}}
\def\funSize#1{\mathrm{size}(#1)}
\def\funAbs#1{\mathrm{abs}(#1)}
\def\t{\texttt{\#}\mathtt{t}}
\def\f{\texttt{\#}\mathtt{f}}
\def\True{\texttt{true}}
\def\False{\texttt{false}}
\def\A{\mathcal{A}}
\def\P{\mathsf{P}}
\def\c{\mathtt{c}}
\def\D{\mathsf{D}}
\def\p{\mathsf{p}}
\def\S{\mathsf{S}}
\def\E{\mathsf{E}}
\def\e{\mathsf{e}}
\def\s{\mathsf{s}}
\def\Z{\mathbb{Z}}
\def\B{\mathbb{B}}
\def\V{\mathbb{V}}
\def\C{\mathsf{C}}
\def\tt{\mathsf{t}}
\def\smap#1#2{\llbracket#1\rrbracket_{#2}}
\def\O{\mathsf{O}}
\def\op{\mathtt{op}}
\def\arity{\mathrm{ary}}
\def\L{\mathsf{L}}
\def\T{\mathsf{T}}
\def\dom{\mathrm{dom}}
\def\Int{\mathsf{Int}}
\def\Bool{\mathsf{Bool}}
\def\semantics#1#2#3{\langle#1,#2\rangle\Downarrow{#3}}
\def\cost#1#2#3#4{\langle#1,#2\rangle\Downarrow_{#4}{#3}}
\def\srule#1#2#3{\frac{#1}{#2}[\Sigma\mathtt{-#3}]}
\def\sref#1{$\Sigma\mathtt{-#1}$}
\def\evalTree#1#2{\mathscr{E}_{\;#1,\,#2}}
\def\tmap#1#2{\{\!|#1|\!\}_{#2}}
\def\l{\ell}
\def\assignable{\mathsf{Asg}}
\def\comparable{{\sim_{\,\T}\;}}
\def\notcmp{\not\sim_{\,\T}\;}
\def\typing#1#2#3#4{#1,#2\vdash#3:#4}
\def\trule#1#2#3{\frac{#1}{#2}[\Gamma\mathtt{-#3}]}
\def\tref#1{$\Gamma\mathtt{-#1}$}
\def\domi#1{\dom_{\mathrm{\,I}}(#1)}
\def\restr#1#2{{(#1)_{\restriction #2}}}
\def\classF{\mathscr{F}}
\def\ic{\mathrm{ic}}
\def\polyn{\mathtt{poly}(n)}
\def\I{\mathscr{I}}
\def\N{\mathbb{N}} 
\section{Introduction}\label{sec:introduction}

Program verification and analysis depend heavily on programming languages.
From the early days of programming in assembly languages to recent experiments in applying deep learning techniques to verification and analysis~\cite{si2020code2inv}, different programming languages have been proposed to address the increasingly demanding concerns from industry on software correctness and safety.
A lot of tools and theoretical frameworks (Hoare logic, symbolic execution, theorem prover, to name a few)~\cite{baier2008principles,ali2017survey,moura2021lean,bertot2013interactive} that have been developed in the field are based on program logics and program semantics.
The fundamental barrier to program verification and analysis is the undecidability.
One of the basic program properties is the termination issue, which is well-known to be undecidable.
Even if a program property is decidable, it is often computationally infeasible.
Exponential blowup is more a rule than an exception.

Modern software systems are composed of many components with complex interconnections.
Different parts of the system may well be designed at different times, in different locations, by different people using different programming languages.
Most components of a software system do not require the full Turing computability.
A good strategy to reduce the burden of verification is to impose strong constraints on the power of the programming language used to develop a particular component, thereby making some properties hold \emph{a priori}.
In many real-world cases, program executions should be feasible, meaning they can run within bounded resources, with running time being the most critical resource.
If we can design a programming language that admits only feasible programs, the issue of program efficiency effectively disappears, placing us in a far better position to pursue program verification and analysis.
To achieve this goal, we must first determine what constitutes \emph{practical efficiency}, on which it is difficult for the industry to reach a clear consensus.

Meanwhile, modeling feasible computation is also a central theme in theoretical research~\cite{hartmanis1965computational}.
The thesis promoted by Cobham~\cite{cobham1965intrinsic} and Edmonds~\cite{edmonds1965paths} is widely accepted among theoretical researchers: an algorithm is considered \emph{theoretically feasible} if it can be computed within polynomial time.
The complexity classes $\classP$ (the class of the problems solvable in polynomial time) for decision problems and $\classFP$ for functions were introduced to capture the theoretically feasible problems.
Strictly speaking, theoretical feasibility cannot fully predict a program’s practical performance, with the Coppersmith-Winograd algorithm~\cite{COPPERSMITH1990251} and the simplex algorithm~\cite{Klee1970HOWGI} for linear programming serving as famous counterexamples.
Nonetheless, the Cobham-Edmonds Thesis still provides a solid explanation for practical efficiency and has inspired extensive related research~\cite{cook2023feasibly,cook1989functional}.
Based on this, designing a simple and elegant programming language that precisely characterizes the class $\classFP$ is a crucial step toward implementing our strategy.

The original characterization of $\classFP$ is in terms of polynomial-time Turing machines that involve low-level manipulations and incur undecidability issues.
A different research avenue is to design light-weight models that characterize $\classFP$ in a direct manner.
The seminal work of Bellantoni and Cook~\cite{bellantoni1992new} has laid down the foundation for this line of research and has inspired a wide range of subsequent developments~\cite{oitavem2022polynomial,hainry2020tier,hainry2023programming}.
This line of thinking has led to the implicit computational complexity (ICC) theory~\cite{dal2010short}, whose goal is to characterize complexity classes without imposing explicit restrictions on resources.
Various research, such as linear logic systems~\cite{girard1987linear,lafont2004soft}, bounded arithmetic~\cite{buss1985bounded,buss1986polynomial,gurevich1983algebras,hainry2021complexityparser}, and the soft lambda calculus~\cite{baillot2004soft}, also falls within this line of work.

With the strong connection to application contexts, programming languages are well-suited to bridge theoretical foundations and practical concerns.
In this light, a program for $\classFP$ represents both a continuation of ICC research and a further exploration of program verification.
The challenge of designing such a language is twofold.
Firstly, does it exist?
Secondly, does such a language look natural?
Studies in ICC~\cite{jones1999logspace,marion2011type,hainry2023general} only partially affirm the first question.
For the language to look natural, the loop statements are preferred over recursion.
The problem now is that one has to settle for the loop statements without any bounds.
Such programs may not terminate, which brings us back to square one.

To design a language that is both intuitive to use and formally guarantees polynomial-time complexity is inherently challenging.
This creates a trade-off, and in this paper, we aim to strike a balance between soundness and naturalness.
We propose a novel programming approach based on an imperative programming language called PolyC that captures precisely the problems in $\classFP$.
The design of PolyC enables pre-verified efficient programming by embedding verification guarantees into the language. 
This contrasts to the practice in which implementation comes first and safety verification comes afterwords, often incurring significant overhead.
In safety-critical scenarios such as aviation, PolyC, which guarantees termination and low time complexity, offers inherent advantages.

Several programming languages and type systems have been designed to characterize complexity classes, including functional languages~\cite{jones1999logspace}, object-oriented languages~\cite{hainry2023general}, tier-based type systems~\cite{marion2011type}, and so on.
It has become a consensus that restrictions on the usage of variables in these languages is necessary. 
Despite sharing this common mechanism with the prior works, our design differs in some critical dimensions.
\begin{enumerate}
    \item Our priority is to design a practical programming language for $\classFP$, rather than proposing yet another model.
    In PolyC, termination and polynomial-time execution are guaranteed for all well-typed programs. 
    Previous models~\cite{marion2011type,hainry2023general} may admit non-terminating programs, necessitating termination verification techniques~\cite{lee2001size,kuwahara2014automatic} to ensure polynomial-time termination.
    \item Most previous works rely on flow analysis~\cite{volpano2000} to ensure non-interference that prohibits information flow across types. 
    In contrast, PolyC restricts the use of variables in different program blocks.
    \item We provide a deterministic static type-checker accompanied by an interpreter to validate practical feasibility. 
    The implementation issue of the type system is rarely discussed in literature.
\end{enumerate}

\subsubsection{Contribution.}
The goal of this paper is to lay down the foundation for the proposed approach.
We shall confine our attention to {PolyC}, focusing on the formal semantics of the language and the equivalence property we have explained in the above.
The main contributions of the paper are summarized as follows:
\begin{enumerate}
    \item  We present a formal description of CorePolyC, an imperative programming language with simple syntax, operational semantics, and a type system. We demonstrate that well-typed CorePolyC programs always terminate. Additionally, we show that the set $\classPC$ of the CorePolyC computable functions is precisely $\classFP$ (Theorem~\ref{thm:main}).
    \item We extend CorePolyC to a more user-friendly programming language, PolyC, and demonstrate that they are equally expressive (Theorem~\ref{thm:poly-main}).
    We have developed an interpreter for PolyC, allowing users to write and run PolyC programs for problem-solving. 
    The details are available \href{https://github.com/AzureSkyChen/polyc}{online}.
    \item We have pointed out that the verification problem for PolyC programs is relatively easy (Proposition~\ref{prop:unfold}), and our type system can be further extended as a tool for time-complexity analysis of other programming languages.
\end{enumerate}

\subsubsection{Organization.}
Section~\ref{sec:CorePolyC} defines CorePolyC.
Some key properties of PolyC programs are stated.
Section~\ref{sec:main_theorem} proves that CorePolyC characterizes $\classFP$.
Section~\ref{sec:polyc} formally introduces PolyC and establishes its equivalence to CorePolyC in terms of expressive power. 
Section~\ref{sec:conclusion} concludes.

\section{Core PolyC}
\label{sec:CorePolyC}

Core PolyC (CorePolyC) is a strongly typed programming language.
A CorePolyC program outputs an integer upon termination.
The semantics and the type system guarantee that all well-typed CorePolyC programs terminate.

\subsection{Syntax}
\label{subsec:syntax}

The syntax of CorePolyC is defined in Fig.~\ref{fig:vocalbulary}.
CorePolyC is a strongly typed language requiring explicit type annotations.
There are three \textit{fundamental data types}: $\typeIterInt$ (for iterable int), $\typeInt$ and $\typeBool$.
The difference between an $\typeIterInt$ variable and an $\typeInt$ variable is that the former must abide by specific rules given in Section~\ref{sec:semantics}.
Notably, in CorePolyC, integers (both $\typeIterInt$ and $\typeInt$) are unbounded, which aligns with standard models of computation used in complexity theory. 
Let $\Int:=\{\typeInt,\typeIterInt\}$ and $\Bool:=\{\typeBool\}$.
We write $\tt$ for an element of $\,\T := \Int\cup\Bool$.
The set of variables is denoted by $\L$, and we typically use $x, y, z, \ldots$ to range over $\L$.
A variable is \textit{iterable} if it is declared as $\typeIterInt$.

\begin{figure}[t]
\centering
\begin{align*}
\texttt{Types\qquad\quad\;\;t}::=&\texttt{ \textbf{iint} | \textbf{int} | \textbf{bool}}\\
\texttt{Operator\quad\,\;op}::=&\texttt{ + | - | / | \% | size |  }\\
&\texttt{ >= | <= | > | < | == | != | ! | \&\& | || }\\
\texttt{Expressions e}::=&\texttt{ $x$ | c | op($\widetilde{\e}$) | (e) }\\
\texttt{Statements\quad s}::=&\texttt{ t $x$; | $x$=e; | \{ $\widetilde{\s}$ \} | \textbf{if}(e) $\s_1$ \textbf{else} $\s_2$} \texttt{| \textbf{for}($x$<\textbf{size}(e)) s}\\
\texttt{Programs\qquad p}::=&\texttt{ \textbf{int} \textbf{main}(\textbf{int} $x_1$,\dots,\textbf{int} $x_m$)\{ $\widetilde{\s}$ \textbf{return} e; \}}
\end{align*}
\caption{The syntax of CorePolyC. }
\label{fig:vocalbulary}
\end{figure}

Let $\D:=\{\texttt{0,\ldots,9}\}$ and $\C :=\D^{+} \cup \{\True,\False\}$ represent all the literals.
The set of all operators in CorePolyC is denoted by $\O$.
It is important to note that the multiplication \texttt{*} is not present in our language, even though it is a fundamental operator in almost all programming languages.
An intuitive explanation of this design will be provided at the end of this section.
All operators and their precedence and associativity are standard except for $\opSize$, which is defined as
the bit-size of $\funAbs{x}$, i.e.,
$
    \funSize{x}=\size{x}=\lceil \log (\funAbs{x}+1) \rceil.
$

Expressions are constructed from variables $x\in \L$, constants $\c\in \C$, and operators $\op\in\O$.
We abbreviate a sequence of expressions $\e_1, \ldots, \e_m$ to $\widetilde{\e}$,
where the meta notation $m$ ranges over the set $\N$ of natural numbers.
An expression is \textit{iterable} if it contains only iterable variables;
it is \textit{non-iterable} otherwise.

\begin{remark}
    In the syntax, all expressions are generated in prefix form. 
    However, for readability, the examples presented later will typically use \emph{infix notation} for binary operators. 
    For instance, we write \texttt{1==1} instead of \texttt{==(1,1)}.
\end{remark}

The statements are standard except for the loop statement.
The notation $\widetilde{\s}$ is for a sequence of statements.
A loop statement includes a \textit{loop counter} and an explicit \textit{loop bound} on the number of iterations, which must be a size expression.
A simple example is $\texttt{\textbf{for}(i<\textbf{size}(x))\{y=y+1;\}}$,
where the loop counter is \texttt{i} and the loop bound is \texttt{\textbf{size}(x)}.
Note that the counter \texttt{i} should be a new variable implicitly declared here.
When executing the loop statement, the counter \texttt{i} is automatically assigned the initial value $0$, and at the same time, the loop bound is evaluated.
The value of \texttt{i} will be increased by $1$ upon the completion of an execution.
Furthermore, a loop statement must meet two restrictions.
\begin{enumerate}
    \item It is not allowed to declare any iterable variables within any loop body.
    \item No assignment to any iterable variables may occur within the loop body.
\end{enumerate}
The semantics and constraints mentioned above will be ensured, respectively, by the semantic rules and typing rules in the following sections.

We call a string generated by the non-terminal $\texttt{program}$ a \textit{well-formed program}, denoted by $\p$.
Note that \emph{all} parameters are of type \texttt{int}.
Since boolean values can be encoded within integers, for simplicity, we do not introduce \texttt{bool} parameters for CorePolyC.
The number of its inputs is denoted by $\arity{(\p)}$.
The set of the variables referred in $\p$ is denoted by $\L(\p)$.
The set of all well-formed programs is denoted by $\P$.
Similarly the set of all well-formed expressions $\e$ (resp. statements $\s$) is denoted by $\E$ (resp. $\S$).
Unless otherwise stated, the syntactic entities in the following text are all well-formed.

\begin{remark}
    We now provide some intuitions behind the design of CorePolyC.
\begin{enumerate}
    \item \textbf{Loop bound.}
    In CorePolyC, we completely prohibit unbounded loop statements, such as \texttt{\textbf{while}(\textbf{true})\{...\}}, which run the risk of non-termination.
    Such a trade-off is acceptable, since the bounded loops together with the size operation mirror widespread programming practices: loops with explicit iteration ranges (e.g., \texttt{\textbf{for}(\textbf{auto} i: vec)\{$\dots$\}} and \texttt{\textbf{for} i \textbf{in} \textbf{range}(n)\{$\dots$\})} dominate array/matrix algorithms and dynamic programming. 
    They also fit well within some programming practices, where explicit resource parameterization, such as specifying security parameters, is standard.
    \item \textbf{The size operator.}
    The loop bound must be a size expression.
    Without this constraint, we would be able to write a loop statement \texttt{\textbf{for}(i<x)\{$\dots$\}}.
    The loop body can be executed $\Theta(2^{\size{x}})$ times (exponential in the bit-size of $x$), which defeats the polynomial time computability restriction we set out.
    \item \textbf{Absence of multiplication.}
    If the multiplication operator \texttt{*} is admitted in CorePolyC, one can design the following code snippet.
    \begin{lstlisting}
    x=2;
    for(i<size(z)) x=x*x;
    \end{lstlisting}
    It effectively implements the function $z\mapsto 2^{2^{\size{z}}}$, which takes exponential time to output this number.
    Therefore, the multiplication operator must be prohibited in our language.
    As a comparison, Example~\ref{example:fast-mul} provides a valid implementation of multiplication in CorePolyC. 
    A more general discussion of operations will be presented in Section~\ref{subsec:language}.
    
    \item \textbf{Iterable variables and expressions.}
    If we do not impose these distinctions between the iterable and the non-iterable variables, the following code would be legitimate.
    \begin{lstlisting}
    for(i<size(z)){
        y=x;
        for(i<size(y)) x=x+x;
    }
    \end{lstlisting}
    The program calculates the function $(x,z) \mapsto 2^{\size{x}(2^{\size{z}}-1)} x$ whose size grows exponentially.
    In this work, we employ a specialized type system (iterable types) to avoid this problematic scenario.
\end{enumerate}
\end{remark}

\subsection{Semantics}
\label{subsec:semantics}
Let $\A$ denote the set of ASCII characters and $\A^{*}$ the set of all finite-length strings over $\A$.
Let $\Z$ be the set of integers $i, j, k, \dots$, $\N$ the set of natural numbers $m, n, \cdots$, and $\B=\{\t,\f\}$ the set of boolean values $b$.
We denote by $\V:=\Z\cup\B$ the set of all possible values $v$ that variables can be assigned during execution.
A tuple of values $(v_1, \ldots, v_m)$ is abbreviated to $\widetilde{v}$ with its size $\size{\widetilde{v}}:=\sum_{i=1}^m\size{ v_i}$.
Accordingly, the input variables of program $\p$ are often abbreviated to $\widetilde{x}$.

The \textit{auxiliary interpretation operator} $\smap{\cdot}{m}$ is defined partially on $\A^{*}$.
Given an operator $\op$ and its arity $\arity(\op)$, we can
interpret it as an $\arity(\op)$-ary total function $\smap{\op}{\arity(\op)}$.
For all operators given in Fig.~\ref{fig:vocalbulary}, the interpretation is straightforward.
Recall that $\smap{\texttt{-}}{1}$ is defined as negation function, and $\smap{\texttt{-}}{2}$ is the subtraction function.
For maintaining the totality of functions, $\smap{\texttt{/}}{2}$ and $\smap{\texttt{\%}}{2}$ are defined to return $0$ when the divisor is $0$.
Conventional language implementations typically handle such situations by throwing a runtime error.
A constant can be treated as a $0$-ary operator as follows.
\begin{equation}
    \smap{\True}{0}=\t\;,\quad \smap{\False}{0}=\f,
\end{equation}
\begin{equation}
    \smap{\texttt{0}}{0}=0,\dots,\smap{\texttt{9}}{0}=9,\quad\smap{w\alpha}{0}=10\smap{w}{0}+\smap{\alpha}{0}\text{ for }w\in\D^{+}\text{ and }\alpha\in \D.
\end{equation}
When applied to an element $\t$ of $\T\subseteq \A^{*}$, the function $\smap{\cdot}{0}$ is defined as:
\begin{equation}
   \smap{\tt}{0}=
\begin{cases}
0\,,\quad&\text{if}\;\tt\in\Int, \\
\f\,,\qquad&\text{if}\;\tt\in\Bool,
\end{cases}
\end{equation}
which gives the initial values for variables of these types.
From now on, we will omit the subscript of $\smap{\cdot}{m}$ whenever there is no ambiguity.

A \textit{store environment} $\Sigma$ is essentially a mapping from $\L$ to $\V$ with a finite domain $\dom(\Sigma)$.
It is a snapshot of the variable values at a certain moment during program execution.
An environment $\Sigma$ can be spelt out as $[x_1\mapsto v_1][x_2\mapsto v_2]\cdots[x_m\mapsto v_m]$ where $\dom(\Sigma)=\{x_1,\dots,x_m\}$.
The empty store environment is denoted by $\emptyset$, and $\Sigma[x\mapsto v]$ is an \textit{update} of $\Sigma$ defined as follows:
\begin{equation}
    \emptyset(x)= \;\uparrow\;,\quad\Sigma[x\mapsto v](y)=\begin{cases}\Sigma(y)\;, &\texttt{if }y\not\equiv x, \\v\;, &\texttt{if }y\equiv x, \end{cases}
\end{equation}
where $\uparrow$ indicates undefinedness and $\equiv$ is the syntactic equality.
For $m$-tuples $\widetilde{x}$ and $\widetilde{v}$, the update $\Sigma[x_1\mapsto v_1]\cdots[x_m\mapsto v_m]$ is abbreviated to $\Sigma[\widetilde{x}\,\mapsto\,\widetilde{v}]$.

\begin{figure*}[t]
    \centering
    \[
    \srule{x\in\dom(\Sigma)}{\semantics{\Sigma}{x}{\Sigma(x)}}{Var}\qquad
    \srule{\c\in\C}{\semantics{\Sigma}{\c}{\smap{\c}{}}}{Const}\qquad
    \srule{\semantics{\Sigma}{\e}{v}}{\semantics{\Sigma}{\texttt{(e)}}{v}}{Paren}
    \]
    \[
    \srule{\semantics{\Sigma}{\widetilde{\e}}{\widetilde{v}}\quad \op\in\O}{\semantics{\Sigma}{\texttt{op($\widetilde{\e}$)}}{\smap{\mathtt{op}}{}(\widetilde{v})}}{Op}\qquad
    \srule{}{\semantics{\Sigma}{\tt\;x\texttt{;}}{\Sigma[x\mapsto{\smap{\tt}{}}]}}{Decl}\]
    \[
    \srule{\semantics{\Sigma}{\;\widetilde{\s}\;}{\Sigma'}}{\semantics{\Sigma}{\texttt{\{ }\widetilde{\s}\texttt{ \}}}{\Sigma'}}{Block}\qquad
    \srule{x\in\dom(\Sigma)\quad \semantics{\Sigma}{\e}{v}}{\semantics{\Sigma}{\texttt{$x$=e;}}{\Sigma[x\mapsto{v}]}}{Asgmt}
    \]
    \[
    \srule{\Sigma_0:=\Sigma\quad\semantics{\Sigma_{i-1}}{\s_i}{\Sigma_{i}} ,\quad\text{for}\;1\le i\le m.}{\semantics{\Sigma}{\;\widetilde{\s}\;}{\Sigma_m}}{Seq}\]
    \[
    \srule{\semantics{\Sigma}{\e}{b}\quad \Sigma':=\begin{cases}\Sigma_1\,,\;&\text{if}\;b\land \semantics{\Sigma}{\s_1}{\Sigma_1}, \\
    \Sigma_2\,,\;&\text{if}\;\lnot b\land \semantics{\Sigma}{\s_2}{\Sigma_2}.\end{cases}}
    {\semantics{\Sigma}{\texttt{\textbf{if}(}\e\texttt{) } \s_1 \texttt{ \textbf{else} } \s_2}{\Sigma'}}{Cond}
    \]
    \[
    \srule{\semantics{\Sigma}{\e}{i}\quad\Sigma_0:=\Sigma\quad\semantics{\Sigma_{j}[x\mapsto j]}{\s}{\Sigma_{j+1}},\;\text{for}\;\;0\le j< i.}{\semantics{\Sigma}{\texttt{\textbf{for}(}x \texttt{<} \e \texttt{) } \s}{\Sigma_{i}}}{Loop}
    \]
    \[
    \srule{\semantics{\Sigma}{\;\widetilde{\s}\;}{\Sigma'}\quad\semantics{\Sigma'}{\e}{i}}{\semantics{\Sigma}{\texttt{\textbf{int} \textbf{main}(}\widetilde{\tt\;x}\texttt{)\{ }\widetilde{\s}\texttt{ \textbf{return} }\e\texttt{;\}}}{i}}{Prog}
    \]
    \caption{Semantics of CorePolyC.}
    \label{fig:semantics}
\end{figure*}

Once we know the current store environment, we can evaluate an expression to a value.
The \emph{semantics} of CorePolyC is defined by the structural operational rules in Fig.~\ref{fig:semantics}.
The judgments $\semantics{\Sigma}{\e}{v}$ specifies that the evaluation of $\e$ terminates with value $v$ under the store environment $\Sigma$.
The judgement $\semantics{\Sigma}{\widetilde{\e}}{\widetilde{v}}$ is an abbreviation for $\semantics{\Sigma}{\e_i}{v_i}$ for $1 \leq i \leq m$, assuming that the length $m$ of the tuple is known.
Similarly, the judgment $\semantics{\Sigma}{\s}{\Sigma'}$ states that under the store environment $\Sigma$, the execution of $\s$ terminates with the effect that updates $\Sigma$ to $\Sigma'$.
The judgement $\semantics{\Sigma}{\p}{v}$ indicates that the execution of program $\p$ terminates with output $v$ under $\Sigma$.

For $w\in\E\cup\S\cup\P$, if $\semantics{\Sigma}{w}{\cdot}$, let the derivation tree of rule instances that lead to this conclusion be denoted by $\evalTree{\Sigma}{w}$.
It is well-defined and unique since the semantics rules are all deterministic.

\subsection{Termination-Guarantee Typing System}
\label{sec:semantics}

In this part, we present the typing rules of CorePolyC, which enforce the special restrictions on iterable variables.
It is expected that the execution of every well-typed CorePolyC program terminates and produces the correct output.

\subsubsection{Typing Rules.}
A \textit{typing environment} $\Gamma:\L\to\T$ is a partial map from variable names to types such that the domain $\dom(\Gamma)$ is finite.
The operations of $\Gamma$ are defined similarly as those of $\Sigma$.
The empty typing environment is also denoted by $\emptyset$.
The equivalence relation $\comparable$ is induced by the partition of $\T$ by $\Int,\Bool$.
In other words, $\tt_1 \comparable \tt_2$ if $\tt_1,\tt_2\in\Int$ or $\tt_1,\tt_2\in\Bool$.

Let $(\Int, \preccurlyeq)$ be a total order relation determined by $\typeIterInt\preccurlyeq\typeInt$.
For $m$-tuple $\widetilde{\tt}\in\Int^m$,  we write $\bigvee_{i=1}^m \tt_i$ for the supremum under the order relation $\preccurlyeq$.
By definition, $\bigvee_{i=1}^m \tt_i\equiv\typeIterInt$ if and only if $\tt_i\equiv\typeIterInt$ for $1\le i\le m$.
The type of an expression is the supremum of the types of the variables in the expression.

The \textit{auxiliary typing operator} $\tmap{\cdot}{m}$ is defined partially on $\A^{*}$. For $\c\in\C$, $\tmap{\c}{0}\in\T$ is defined as the type of $\c$. Thus,
\begin{equation}
    \tmap{\True}{0}=\tmap{\False}{0}=\typeBool,\qquad \tmap{\texttt{c}}{0}=\typeIterInt,\quad\text{for }\c\in\D^{+}.
\end{equation}
For an $m$-ary operator $\op\in\O$, $\tmap{\op}{m}$ is a partial map from $\T^{m}$ to $\T$: 
\begin{itemize}
    \item $\tmap{\op}{m}(\typeBool^{m})=\typeBool$ if $\op\in\{\texttt{\&\&},\texttt{||},\texttt{!}\}$;
    \item $\tmap{\op}{m}(\widetilde{\tt})=\typeBool$ if $\op\in\{\texttt{>=},\texttt{<=},\texttt{>},\texttt{<},\texttt{==},\texttt{!=}\}$ and $\widetilde{\tt}\in\Int^m$;
    \item $\tmap{\op}{m}(\widetilde{\tt})=\bigvee_{i=1}^{m}t_i$ if $\op\in\{\texttt{+},\texttt{-},\texttt{/},\texttt{\%}\}$ and $\widetilde{\tt}\in\Int^m$;
    \item $\tmap{\op}{1}(\typeIterInt)=\typeIterInt$ if $\op\equiv\opSize$.
\end{itemize}
The subscript $m$ in $\tmap{\cdot}{m}$ is often omitted.

A \textit{loop indicator} $\l$ is a boolean variable indicating whether the current statement is inside any loop body.
The predicate $\assignable$ is defined on $\B\times\T$ by $\assignable(\l\,, \tt)=\lnot (\l\land \tt\equiv\typeIterInt)$, meaning that an assignment can be performed whenever it is not inside a loop or the variable being assigned is not of iterable type.

\begin{figure*}[t]
    \centering
    \[
    \trule{x\in\dom(\Gamma)}{\typing{\Gamma}{\l}{x}{\Gamma(x)}}{Var}\qquad
    \trule{\c\in\C}{\typing{\Gamma}{\l}{\c}{\tmap{\c}{}}}{Const}\qquad
    \trule{\typing{\Gamma}{\l}{\e}{\tt}}{\typing{\Gamma}{\l}{\texttt{(e)}}{\tt}}{Paren}\]
    \[
    \trule{\typing{\Gamma}{\l}{\widetilde{\e}}{\widetilde{\tt}}\quad \widetilde{\tt}\in\dom(\tmap{\mathtt{op}}{})\quad \mathtt{op}\in\O}{\typing{\Gamma}{\l}{\texttt{op($\widetilde{\e}$)}}{\tmap{\mathtt{op}}{}(\widetilde{\tt})}}{Op}
    \qquad
    \trule{\assignable(\l\,,\tt)\quad x\not\in\dom(\Gamma)}{\typing{\Gamma}{\l}{\tt\;x\texttt{;}}{\Gamma[x\mapsto \tt]}}{Decl}
    \]
    \[
    \trule
    {\typing{\Gamma}{\l}{\;\widetilde{\s}\;}{\Gamma'}}
    {\typing{\Gamma}{\l}{\texttt{\{ }\widetilde{\s}\texttt{ \}}}{\Gamma}}
    {Block}
    \qquad
    \trule
    {\typing{\assignable(\l\,,\Gamma(x))\quad\Gamma}{\l}{\e}{\tt}\quad\Gamma(x)\comparable \tt}
    {\typing{\Gamma}{\l}{x\texttt{=e;}}{\Gamma}}
    {Asgmt}\]
    \[
    \trule
    {\Gamma_0:=\Gamma\quad\typing{\Gamma_{i-1}}{\l}{\s_i}{\Gamma_{i}} ,\quad\text{for}\;1\le i\le m.}
    {\typing{\Gamma}{\l}{\;\widetilde{\s}\;}{\Gamma_m}}
    {Seq}
    \]
    \[\trule
    {\typing{\Gamma}{\l}{\e}{\typeBool}\quad \typing{\Gamma}{\l}{\s_1}{\Gamma_1}\quad \typing{\Gamma}{\l}{\s_2}{\Gamma_2}}
    {\typing{\Gamma}{\l}{\texttt{\textbf{if}(}\e\texttt{) } \s_1 \texttt{ \textbf{else} } \s_2}{\Gamma}}
    {Cond}\]
    \[
    \trule
    {\typing{\Gamma}{\l}{\e}{\typeIterInt}\quad x\not\in\dom(\Gamma)\quad\typing{\Gamma[x\mapsto{\tt}]\,}{\t}{\s}{\Gamma'}}
    {\typing{\Gamma}{\l}{\texttt{\textbf{for}(}x \texttt{<} \e \texttt{) } \s}{\Gamma}}
    {Loop}\]
    \[
    \trule
    {\typing{\Gamma[\;\widetilde{x}\mapsto\widetilde{\tt}\;]\,}{\l}{\;\widetilde{\s}\;}{\Gamma'}\quad\typing{\Gamma'}{\l}{\e}{\tt'\in\Int}}
    {\typing{\Gamma}{\l}{\texttt{\textbf{int main}(}\widetilde{\tt\;x}\texttt{)\{ }\widetilde{\s}\texttt{ \textbf{return} }\e\texttt{;\}}}{\typeInt}}
    {Prog}\]
    \caption{The typing rules of CorePolyC.}
    \label{fig:typing}
\end{figure*}

The typing judgment of the form $\typing{\Gamma}{\l}{\e}{\tt}$ states that under the typing environment $\Gamma$ and the loop indicator $\l$, the expression $\e$ has type $\tt$.
Similarly, we use typing judgments of the form $\typing{\Gamma}{\l}{\s}{\Gamma'}$ and $\typing{\Gamma}{\l}{\p}{\tt}$ for statements and programs, respectively.
Fig.~\ref{fig:typing} gives the typing rules for CorePolyC.

We can now formally state that an expression $\e$ is iterable under the typing environment $\Gamma$ if $\typing{\Gamma}{\f}{\e}{\typeIterInt}$ holds.
The definition of $\tmap{\cdot}{}$, together with the rule~\tref{Op} ensure that the operand of the $\opSize$ operation must be iterable.
Note that among these rules, the only place where $\l$ is modified is in the premise of \tref{Loop}. Setting $\l$ to $\t$ indicates that type checking is performed within a loop body.
On the other hand, $\l$ is used in rules \tref{Decl} and \tref{Asgmt}, which formalizes the restrictions on the iterable variables by $\assignable$ (See Example~\ref{exp:typing_derivation} and Example~\ref{exp:loop} for an intuitive explanation).
Proposition \ref{prop:invariant} in the next section provides a formal statement of these observations.

We end this part with some examples to explain our type system.
\begin{example}[Fast multiplication]
\label{example:fast-mul}
 Fig.~\ref{fig:mul} gives two programs in the style of CorePolyC that intend to implement the fast multiplication of two positive integer variables \texttt{x}, \texttt{y}.
\begin{figure}[t]
\begin{minipage}{0.45\linewidth}
\centering
\begin{lstlisting}[xleftmargin=1em,xrightmargin=1em]
int main(int x,int y){
    iint o;
    for(i<size(y)){
        if(y%2!=0){
            o=o+x;
        } else {}
        x=x+x;
        y=y/2;
    }
    return o;
}
\end{lstlisting}
\end{minipage}
\hspace{10mm}
\begin{minipage}{0.45\linewidth}
\centering
\begin{lstlisting}[xleftmargin=1em,xrightmargin=1em]
int main(int x,int y){
    int o; iint z; z=y;
    for(i<size(z)){
        if(y%2!=0){
            o=o+x;
        } else {}
        x=x+x;
        y=y/2;
    }
    return o;
}
\end{lstlisting}
\end{minipage}
\caption{Two CorePolyC implementations of fast multiplication. 
The left one violates the rules~\tref{Loop} and~\tref{Asgmt}, while the right one is well-typed.}
\label{fig:mul}
\end{figure}
Both programs are  well-formed but only the right one is well-typed. The left program violates the rules~\tref{Loop} and~\tref{Asgmt}, because the non-iterable variable \texttt{y} appears as the parameter of \texttt{size} operation, and the iterable variable \texttt{o} is modified in the loop body.

Note that the program on the right declares an iterable variable $z$.
Therefore, this part of the code cannot be placed inside a loop to achieve the effect of \texttt{\textbf{for}(i<\textbf{size}(z)) x=x*x;}, which, as we have explained earlier, should be prohibited.
This is also equivalent to defining a function \texttt{mul(x, y)} and attempting to invoke it within a loop.
In Section~\ref{subsec:language}, we will discuss how to introduce function calls in PolyC to prevent such constructions.
\end{example}

\begin{example}[Typing derivation]
\label{exp:typing_derivation}
Considering the statement $\s_1\equiv\;$\texttt{x=x+x;} and $\Gamma_1=[\texttt{x}\mapsto\typeInt]$, we may derive that $\typing{\Gamma_1}{\t}{\s_1}{\Gamma_1}$.
The corresponding typing derivation is shown as follows:
$$
{
\inference
{
    \assignable(\t,\typeInt)
    \enspace
    {
        \inference
        {
            {
            \inference
            {
                \texttt{x}\in\dom(\Gamma_1)
            }
            {
                \typing{\Gamma_1}{\t}{\texttt{x}}{\typeInt}
            }
            }
            \quad
            (\typeInt,\typeInt)\in\dom(\tmap{\texttt{+}}{})
        }
        {
            \typing{\Gamma_1}{\t}{\texttt{x+x}}{\tmap{\texttt{+}}{}(\typeInt,\typeInt)=\typeInt}\;
        }
    }
    \enspace
    \typeInt\comparable\typeInt
}
{
    \typing{\Gamma_1}{\t}{\s_1}{\Gamma_1}
}
}
$$
where $\tmap{\texttt{+}}{}(\tt_1,\tt_2)=\tt_1\lor \tt_2$.
The above results indicate that if the variable \texttt{x} is of type $\typeInt$, then assignment statements like $\s_1$ can appear within a loop body.
However, if we let $\Gamma_1'=[\texttt{x}\mapsto\typeIterInt]$, the similar type inference cannot be completed because the predicate $\assignable(\t,\Gamma_1'(\texttt{x}))$ evaluates to false.
This implies that if the variable \texttt{x} is of type $\typeIterInt$, then assignments to it cannot occur within any loop body.
The situation with declaration statements is similar.

Furthermore, both $\typing{\Gamma_1}{\f}{\s_1}{\Gamma_1}$ and $\typing{\Gamma_1'}{\f}{\s_1}{\Gamma_1'}$ hold because $\l=\f$ indicates that the current check is carried out in the scenario where $\s_1$ is not within any loop.
In this case, there are no restrictions on assignment or declaration statements.
This example elucidates how $\l$ and $\assignable$ ensure that the program complies with the requirements during type checking.
\end{example}

\begin{example}[Loop]
\label{exp:loop}
Let $\s_2:=\;$\texttt{\textbf{for}(i<\textbf{size}(z)) $\s_1$;} and $\Gamma_2 = \Gamma_1[\texttt{z}\mapsto\typeIterInt]$, with $\Gamma_1$ and $\s_1$ defined in Example \ref{exp:typing_derivation}.
Now, we have $\typing{\Gamma_2}{\f}{\s_2}{\Gamma_2}$ as follows:
{
\small
$$
\inference
{
    {
    \inference
    {
        {
        \inference
        {\texttt{z}\in\dom(\Gamma_2)}
        {\typing{\Gamma_2}{\f}{\texttt{z}}{\typeIterInt}}
        }\;
        \typeIterInt\in\dom(\tmap{\opSize}{})
    }
    {
        \typing{\Gamma_2}{\f}{\texttt{\textbf{size}(z)}}{\typeIterInt}
    }
    }
    \enspace
    \texttt{i}\not\in\dom(\Gamma_2)
    \enspace
    {
    \inference
    {
        \text{$\cdots$}
    }
    {
        \typing{\Gamma_2'}{\t}{\s_1}{\Gamma_2'}
    }
    }
}
{\typing{\Gamma_2}{\f}{\s_2}{\Gamma_2}}
$$
}

\noindent
where $\Gamma_2'=\Gamma_2[\texttt{i}\mapsto\typeIterInt]$.
Since $\Gamma_2'(\texttt{x}) = \Gamma_1(\texttt{x})$, we can obtain $\typing{\Gamma_2'}{\t}{\s_1}{\Gamma_2'}$ similar to Example \ref{exp:typing_derivation}.
Therefore, we have omitted the corresponding part here.
\end{example}

\subsubsection{Termination-Guarantee of Well-Typed Programs.}
The \emph{well-typedness} of syntactical entities with regards to the type system is defined as follows.
\begin{enumerate}
    \item An expression $\e$ is \emph{well-typed} if $\typing{\Gamma}{\l}{\e}{\tt}$ for some $\Gamma$, $\l$ and $\tt$.
    \item A statement $\s$ is \emph{well-typed} if $\typing{\Gamma}{\l}{\s}{\Gamma'}$ for some $\Gamma$, $\l$ and $\Gamma'$.
    \item A program $\p$ is \emph{well-typed} if there exists a type $t$ such that $\typing{\emptyset}{\f}{\p}{\tt}$.
\end{enumerate}
It follows from the structural definition that if $\p$ is well-typed, all syntactically well-formed components within $\p$ are also well-typed.

Now, we establish the connection between types and semantics.
We say that value $v$ is \textit{consistent} with type $\tt$ if $(v,\tt)\in\Z\times\Int\cup\B\times\Bool$.
A store environment $\Sigma$ and a typing environment $\Gamma$ are consistent if $\dom(\Gamma)\subseteq\dom(\Sigma)$, and for any $x\in\text{dom}(\Gamma)$, $\Sigma(x)$ and $\Gamma(x)$ are consistent.
We also say that these two environments $(\Gamma, \Sigma)$ form a consistent pair.
Lemma~\ref{lem:map} formalizes that, in CorePolyC, all constants and operators have consistent types and semantics.

\begin{lemma}\label{lem:map}
The following statements are valid.
\begin{enumerate}
    \item For each $\c\in\C$, $\smap{\c}{}$ and $\tmap{\c}{}$ are consistent.
    \item For each $\tt\in\T$, $\smap{\tt}{}$ and $\tt$ are consistent.
    \item For $\op\in\O$, $\,\widetilde{v}\in\dom(\smap{\op}{})$ and $\;\widetilde{\tt}\in\dom(\tmap{\op}{})$, if $\,\widetilde{v}\,$ are consistent with $\widetilde{\tt}$, then $\smap{\op}{}(\,\widetilde{v}\,)$ are consistent with $\tmap{\op}{}(\,\widetilde{\tt}\,)$.
\end{enumerate}
\end{lemma}

Let $\left(\dom(\Gamma_1)\subseteq\dom(\Gamma_2)\right) \wedge \left(\forall x\in\dom(\Gamma_1). \Gamma_1(x)=\Gamma_2(x)\right)$ be abbreviated to $\Gamma_1\subseteq\Gamma_2$.
That is, all variables in $\Gamma_1$ are also contained in $\Gamma_2$ with the same types.
It can be proved that the consistency described above is preserved throughout the type-checking and evaluation processes of expressions and statements.

\begin{lemma}
Suppose $(\Gamma,\Sigma)$ is a consistent pair and $\l$ is a loop indicator.
\begin{enumerate}
    \item For each expression $\e$, if $\typing{\Gamma}{\l}{\e}{\tt}$ for some $\tt$, there exists a unique value $v\in\V$ such that $\semantics{\Sigma}{\e}{v}$ and $v$ is consistent with $\tt$.
    \item For each statement $\s$, if $\typing{\Gamma}{\l}{\s}{\Gamma'}$ for some $\Gamma'$, then $\Gamma\subseteq\Gamma'$ and there exists a unique store $\Sigma'$ consistent with $\Gamma'$ such that $\semantics{\Sigma}{\s}{\Sigma'}$.
\end{enumerate}
\label{lem:preservation}
\end{lemma}

Recall that all programs have only integer inputs.
We now present the main theorem of this section, which states that, given a set of inputs, any CorePolyC program will terminate and produce a unique integer output.

\begin{theorem}[Type Safety]
    Given a program $\p$ whose input is an $m$-tuple $\widetilde{x}$, for any $\widetilde{v}\in\Z^m$, there exists a unique integer $i$ such that $\semantics{\emptyset[\widetilde{x}\mapsto{\widetilde{v}}]}{\p}{i}$.
    \label{thm:safety}
\end{theorem}

Given two consistent store and typing environments $\Sigma$ and $\Gamma$, let $\domi{\Gamma}$ denote the set $\{x \in \dom(\Gamma) : \Gamma(x) \equiv \typeIterInt\}$.
The \textit{iterable restriction} of $\Sigma$ on $\Gamma$ is defined as $\restr{\Sigma}{\Gamma} := \Sigma \restriction \domi{\Gamma}$.
Two store environments $\Sigma_1$ and $\Sigma_2$ are equivalent, denoted as $\Sigma_1 \sim \Sigma_2$, if  $\dom(\Sigma_1) = \dom(\Sigma_2)$, and $\Sigma_1(x) = \Sigma_2(x)$ for every $x \in \dom(\Sigma_1)$.
The behaviors of the iterable variables are made predictable by the strong constraints imposed on their usage.
Proposition~\ref{prop:invariant} shows that our type system ensures that the iterable variables are well-behaved.

\begin{proposition}[Loop Invariant]
\label{prop:invariant}
Given a statement $\s$ and a consistent pair $(\Sigma,\Gamma)$, if there exists $\Gamma'$ such that $\typing{\Gamma}{\t}{\s}{\Gamma'}$, then
\begin{enumerate}
    \item $\domi{\Gamma'}=\domi{\Gamma}$, and
    \item there exists $\Sigma'$ such that $\semantics{\Sigma}{\s}{\Sigma'}$ and $\restr{\Sigma'}{\Gamma'}\sim\restr{\Sigma'}{\Gamma}\sim\restr{\Sigma}{\Gamma}$.
\end{enumerate}
These two conclusions above correspond respectively to i) no new iterable variables are declared, and ii) no existing iterable variables have their values changed.
\end{proposition}

\section{Core PolyC is a Model of $\classFP$}\label{sec:main_theorem}

This section substantiates our claim that CorePolyC captures precisely the polynomial time computability.
Unless otherwise stated, all programs will be well-typed CorePolyC programs.
The set of all finite-arity total functions is denoted by $\classF$.
Given an $m$-ary program $\p$ with input variables $\widetilde{x}$, for any $\widetilde{v}\in\Z^{m}$, we define $\smap{\p}{}{(\widetilde{v})}$ as the value $v'$ such that $\semantics{\emptyset[\widetilde{x}\mapsto\widetilde{v}]}{\p}{v'}$.
By Theorem \ref{thm:safety}, $\smap{\p}{}$ is a well-defined total function from $\Z^{m}$ to $\Z$, i.e., $\smap{\p}{}\in\classF$.
We call an $f\in \classF$ to be  \textit{CorePolyC Computable} if there is a CorePolyC program $\p$ such that $\smap{\p}{}=f$.
Let $\classPC$ denote the set of all CorePolyC computable functions.
The following theorem is the main result of the paper.
\begin{theorem}[Main Theorem]\label{thm:main}
$\classPC=\classFP$.
\end{theorem}

\begin{figure}[t]
    \centering
\usetikzlibrary{arrows.meta, positioning, decorations.pathreplacing}
\begin{tikzpicture}[
  theorem/.style={align=center,font=\fontsize{6.8pt}{8pt}\selectfont},
  lemma/.style={align=center,font=\fontsize{6.5pt}{8pt}\selectfont},
  arrow/.style={->,> = latex'},
]

\node[theorem] (thm2) at (0,0) {Theorem~\ref{thm:main}\\(Main Theorem)};
\node[theorem] (thm3) at (2.1,0.8) {Theorem~\ref{thm:soundness}\\(Soundness)};
\node[theorem] (thm4) at (2.2,-0.8) {Theorem~\ref{thm:completeness}\\(Completeness)};
\draw [decorate, decoration={brace, amplitude=4pt}] (1.2,-0.8) -- (1.2,0.8);

\node[lemma] (lemma4) at (4,1.4) {Lemma~\ref{lem:class}\\(Time)};
\node[lemma] (lemma3) at (6.3,1.4) {Lemma~\ref{lem:time-size}\\(Time vs. Space)};
\draw[arrow] (lemma3) -- (lemma4);

\node[lemma] (prop2) at (4,0.2) {Proposition~\ref{prop:prog-size}\\(Programs)};
\node[lemma] (lemma7) at (6,0.2) {Lemma~\ref{lem:stmt-invariant}\\(Statements)};
\node[lemma] (lemma6) at (8,0.2) {Lemma~\ref{lem:expr-size}\\(Expressions)};
\node[lemma] (lemma5) at (10,0.2) {Lemma~\ref{lem:op-size}\\(Operators)};

\draw[arrow] (lemma5) -- (lemma6);
\draw[arrow] (lemma6) -- (lemma7);
\draw[arrow] (lemma7) -- (prop2);

\draw [decorate, decoration={brace, amplitude=4pt}]
  (3.1,0.2) -- (3.1,1.4);

\definecolor{blue}{HTML}{191A94}
\draw[blue, dashed, opacity=1] (3.16,-0.2) rectangle (11,0.6);
\node[theorem,text=blue] (ex1) at (9.1,-0.6) {\textbf{Branch II:} Analysis of the\\output size step by step};

\definecolor{cyan}{HTML}{04579B}
\draw[cyan, dashed, opacity=1] (3.16,1) rectangle (11,1.8);
\node[theorem,text=cyan] (ex1) at (9.1,1.4) {\textbf{Branch I:} Reduction\\from time to space};

\node[lemma] (lemma9) at (4.6,-0.8) {Lemma~\ref{lem:simulate}\\(Simulation)};
\node[lemma] (lemma8) at (6.6,-0.8) {Lemma ~\ref{lemma:poly-con}\\(Clock)};

\draw[arrow] (lemma8) -- (lemma9);
\draw[arrow] (lemma9) -- (thm4);
\definecolor{gray}{HTML}{0E600F}
\draw[gray, dashed, opacity=1] (3.7,-1.2) rectangle (7.4,-0.4);
\node[theorem,text=gray] (ex1) at (6.5,-1.4) {Simulation of polynomial-time Turing machines};

\end{tikzpicture}

    \caption{The dependency graph for the proof of Theorem~\ref{thm:main}. }
    \label{fig:dependency_graph}
\end{figure}

This section is devoted to the proof of Theorem~\ref{thm:main}.
The overall proof is divided into two main parts, as illustrated in Fig.~\ref{fig:dependency_graph}:
\begin{enumerate}
    \item In Section~\ref{subsec:soundness}, we establish the soundness direction $\classPC \subseteq \classFP$ (Theorem~\ref{thm:soundness}), i.e., all CorePolyC programs run in polynomial time.
   Since all operators are interpreted as polynomial-time computable functions, it suffices to show that (i) every well-typed CorePolyC program performs only a polynomial number of operations, and (ii) each operation is applied to operands whose sizes are polynomially bounded.
   Interestingly, due to the well-structured nature of CorePolyC programs, (i) can be reduced to the proof of (ii) (Branch I).
   Therefore, it is sufficient to analyze how variable sizes grow step by step, from operators to expressions, statements, and finally programs (Branch II).
   \item In Section~\ref{sec-Completeness}, we prove the completeness direction $\classFP \subseteq \classPC$ (Theorem~\ref{thm:completeness}).
   We show that CorePolyC programs can construct a “hard-wired clock” for any polynomial, enabling it to simulate each step of a polynomial-time Turing machine, and ultimately produce a consistent output.
\end{enumerate}

\subsection{Soundness}\label{subsec:soundness}

To analyze programs' running time, we introduce the cost semantics for CorePolyC in Fig.~\ref{fig:cost_semantics}.
The judgments are in the form $\cost{\Sigma}{\cdot}{\cdot}{k}$, where $k$ bounds the number of steps of evaluation.
For example, the cost semantic rule for iteration statements is defined as follows.
\[
\srule{\cost{\Sigma}{\e}{i}{k}\quad\Sigma_0:=\Sigma\quad\cost{\Sigma_{j}[x\mapsto j]}{\s}{\Sigma_{j+1}}{k_j},\;\text{for}\;\;0\le j< i.}{\cost{\Sigma}{\texttt{\textbf{for}(x<e) s}}{\Sigma_{i}}{k+\sum_{j=0}^{i-1}k_j}}{Loop}
\]
This rule intuitively means that the running time of executing a loop is the sum of the running time of each iteration.
The complete set of rules is defined in Appendix~\ref{app:cost_semantics}.
Given a program $\p$,  there is a unique $k$ such that $\cost{\emptyset[\;\widetilde{x}\mapsto\widetilde{v}\;]}{\p}{v'}{k}$.
We define the \textit{instruction count} of $\p$ under input $\widetilde{v}$ as $\ic(\p,\widetilde{v}):=k$, which is essentially the running time of $\p$.

\subsubsection{Time versus Space in CorePolyC.}
The function $\size{\cdot}$ can be extended from $\Z$ to $\V$ by defining $\size{\t}=\size{\f}=1$.
Given a store environment $\Sigma$, the notation $\size{\Sigma}$ stands for $\max_{x\in\dom(\Sigma)}\size{\Sigma(x)}$, that is, the maximum length of the values currently stored in $\Sigma$.
Given a derivation tree $\evalTree{\Sigma}{p}$, the size $\size{\evalTree{\Sigma}{p}}$ is defined as
$\max \{\size{\Sigma'}:\Sigma'\text{ appears in }\evalTree{\Sigma}{p}\}$.
Let $T:\N\to\N$ denote a time function with $T(n)\ge n$.
Suppose the size of an input $\;\widetilde{v}$ to a program $\p$ is $n$.
We say that
\begin{enumerate}
    \item  $\p$ \textit{has an instruction count of} $O(T(n))$ if $\ic(\p, \widetilde{v}) = O(T(n))$,
    \item  $\p$ \textit{produces output of size} $O(T(n))$ if $\size{\smap{\p}{}(\widetilde{v})}=O(T(n))$, and
    \item $\p$ \textit{generates (intermediate) values of size} $O(T(n))$ if $\size{\evalTree{\emptyset[\;\widetilde{x}\mapsto\widetilde{v}\;]}{\p}}=O(T(n))$.
\end{enumerate}
These parameters measure the time/space resources consumed by $\p$.
The following crucial lemma of Branch I in Fig.~\ref{fig:dependency_graph} demonstrates that the time complexity and space complexity of CorePolyC programs are strongly correlated.
\begin{lemma}\label{lem:time-size}
Let $T(n)\ge n$ be a time function.
The following three are equivalent:
\begin{enumerate}
    \item All CorePolyC programs have an instruction count of $O(T(n))$.
    \item All CorePolyC programs only produce outputs of size $O(T(n))$.
    \item All CorePolyC programs only generate values of size $O(T(n))$.
\end{enumerate}
\end{lemma}
We can now define the \textit{complexity class} $\classPC[T]$ as the set of all the functions in $\classPC$ computable by CorePolyC programs with an instruction count of $O(T(n))$.
The following is an immediate consequence of the above lemma.

\begin{corollary}
\label{co:time-class}
If $\,\size{f(\widetilde{v})}=O(T(n))$ for all $f\in\classPC$, then $\classPC=\classPC[T]$.
\end{corollary}

It should be clear that Lemma~\ref{lem:time-size} still holds if we replace $T(n)$ by a family of functions $\polyn$, say the family of polynomials.
Lemma~\ref{lem:class} reveals the relationship between Corollary~\ref{co:time-class} and the soundness of CorePolyC.

\begin{lemma}\label{lem:class}
$\classPC\subseteq \classFP$ if and only if $\forall\,f\,{\in}\,\classPC.\,\size{f(\widetilde{v})}=O(\polyn)$.
\end{lemma}
\begin{proof}
Assume that each $f\in\classPC$ can be computable in polynomial time by a Turing machine.
Then its output size can only be polynomial.
Hence, the $\implies$ holds.
As for the $\impliedby$ direction, assume that $\forall\,{f}\,\in\classPC.\,\size{f(\widetilde{v})}=O(\polyn)$ holds.
Then the lengths of values generated by the CorePolyC programs are all polynomial according to Lemma \ref{lem:time-size}.
For any $f\in\classPC$, we also have that $f\in\classPC\subseteq\bigcup_{k\ge0}\classPC{[n^{k}]}$ by Corollary \ref{co:time-class}.
Thus, there exists a program $\p$ computing $f$, and $\p$ can be evaluated in polynomial steps.
We simulate $\p$ instruction by instruction using a Turing machine $M$. Since $\p$ has only a polynomial number of polynomial-time computable operations, and each operation only involves operands of polynomial size, the simulation can be completed in polynomial time.
Therefore $f\in\classFP$, from which $\classPC\subseteq\classFP$ follows.
\qed
\end{proof}

\subsubsection{Polynomial Bounds for Output.}\label{sec-variable-size}
Lemma~\ref{lem:class} reveals that if CorePolyC can only produce output of polynomial size, the soundness must be valid.
Therefore, we need to examine the size of all outputs produced by CorePolyC programs.
The most fundamental task is to examine the growth rate of values under the operators provided by CorePolyC.
Given an operator $\op\in\O$, we observe that
\begin{enumerate}
    \item $\size{\smap{\op}{}(\widetilde{v})}=1$, if $\op\in\{\texttt{\&\&},\texttt{||},\texttt{!}\}\cup\{\texttt{>=},\texttt{<=},\texttt{>},\texttt{<},\texttt{==},\texttt{!=}\}$,
    \item $\size{\smap{\op}{}(\widetilde{v})}\le \max_{i=1}^{m} \size{v_i}+1$, if $\op\in\{\texttt{+},\texttt{-},\texttt{/},\texttt{\%}\}$, and
    \item $\size{\smap{\opSize}{}{(v)}}=\size{\size{v}}\le\size{v}$, where $\size{\size{v}}$ is the size of $\size{v}$.
\end{enumerate}
A parameter is said to be a constant if it is independent of any input and any store environment.
With this remark, the next lemma should be clear.
\begin{lemma}
\label{lem:op-size}
For any $\op\in\O$, $\size{\smap{\op}{}(\widetilde{v})}\le\max_{i=1}^{m} \size{v_i}+k$, where $k$ is a constant.
\end{lemma}

An expression can only contain a constant number of operators.
Intuitively, the value of an expression should also have a similar upper bound.
Lemma \ref{lem:expr-size} formalizes these ideas.
Note that statement (\ref{lem:expr-size:2}) corresponds to the cases of iterable expressions, whose size depends only on iterable variables in the environment.

\begin{lemma}
\label{lem:expr-size}

Given an expression $\e$ and a consistent pair $(\Gamma,\Sigma)$, if there exist $v$ and $\tt$ such that $\semantics{\Sigma}{\e}{v}$ and $\typing{\Gamma}{\f}{\e}{\tt}$, then the following statements hold.
\begin{enumerate}
    \item $\size{v}\le\size{\Sigma}+k$, where $k$ is a constant.
    \item\label{lem:expr-size:2} If $\e$ is iterable, i.e., $\tt=\typeIterInt$, then $\size{v}\le\size{\restr{\Sigma}{\Gamma}}+k$, where $k$ is a constant.
\end{enumerate}
\end{lemma}

Moreover, let $a \dot{-} b = \max(a - b, 0)$.
Then Lemma~\ref{lem:stmt-invariant} is convenient in upper bound estimation of values after executing a statement.
Particularly, if a statement is inside a loop ($\ell=\t$), then the growth of all values during execution of this statement depends only on the iterable variables.

\begin{lemma}
\label{lem:stmt-invariant}
Given a statement $\s$ and a consistent pair $(\Gamma, \Sigma)$ such that $\typing{\Gamma}{\l}{\s}{\Gamma'}$ and $\semantics{\Sigma}{\s}{\Sigma'}$ for some $\l$, $\Gamma'$ and $\Sigma'$, then the following statements hold.
\begin{enumerate}
    \item If $\l=\t$, then $\size{\Sigma'}\dot{-}\size{\Sigma}=O(\size{\restr{\Sigma}{\Gamma}}^d)$, where $d$ is a constant.
    \item If $\l=\f$, then $\size{\Sigma'}\dot{-}\size{\Sigma}= O(\size{\Sigma}^{d^{k}})$, where $d$ and $k$ are constants.
\end{enumerate}
\end{lemma}
Proofs for the above two lemmas are carried out by structural induction.
Details are given in   \ref{App_proof_of_lemma:expr-size} and \ref{App_proof_of_lemma:stmt-invariant}, respectively.
Intuitively, in Lemma \ref{lem:stmt-invariant},
$d$ represents the maximum depth of the nested loops in the statements, and $k$ represents the number of the sequential loop statements in $\s$, i.e., the number of maximal subtrees induced by \tref{Loop} in the derivation tree.
We are now able to bound the output produced by any programs and deduce the soundness of CorePolyC.

\begin{proposition}\label{prop:prog-size}
Given a program $\p$, it holds that $\size{\smap{\p}{}(\widetilde{v})}=O(\polyn)$.
\end{proposition}

\begin{proof}
Assume that $\p\equiv \texttt{\textbf{int main}(}\tt_1\,x_1,\dots,\tt_m\,x_m\texttt{)\{ }\widetilde{\s}\;\texttt{ \textbf{return} e;\}}$.
There exists a store $\Sigma'$ such that $\semantics{\emptyset[\widetilde{x}\mapsto\widetilde{v}]}{\widetilde{\s}}{\Sigma'}$ and $\semantics{\Sigma'}{\e}{\smap{\p}{}(\widetilde{v})}$.
Since $\size{\emptyset[\widetilde{x}\mapsto\widetilde{v}]}\le \size{\widetilde{v}}$, by Lemma \ref{lem:stmt-invariant} and Lemma~\ref{lem:expr-size}, there exists constants $k$ and $k'$ such that
$\size{\smap{\p}{}(\widetilde{v})}\le \size{\Sigma'}+k'\le\size{\widetilde{v}}+\size{\widetilde{v}}^{d^{k}}+k'= O(\polyn)$.
\qed
\end{proof}

\begin{theorem}[Soundness]
\label{thm:soundness}
$\classPC\subseteq\classFP$.
\end{theorem}
\begin{proof}
    For any $f\in\classPC$, there exists a program $\p$ computing $f$.
    By Proposition \ref{prop:prog-size}, we have $\size{f(\widetilde{v})}=\size{\smap{\p}{}(\widetilde{v})}=O(\polyn)$.
    The result follows from Lemma~\ref{lem:class}.
    \qed
\end{proof}

\subsection{Completeness}\label{sec-Completeness}
We now prove that every function computable in polynomial time by a Turing machine is also computable by a CorePolyC program.
Equivalently, we show that every polynomial-time Turing machine $M$ can be simulated by a CorePolyC program.
During this process, it is standard to construct a variable to indicate the current step of the execution.
Since the running time of $M$ is always bounded by a polynomial $p$, we must prove that CorePolyC programs can generate a value $v$ such that $p(n)$ can be bounded by $\size{v}$.
Lemma~\ref{lemma:poly-con} is sufficient to address the above issue. See Appendix \ref{App_proof_of_lemma:poly-con} for a proof.
\begin{lemma}
\label{lemma:poly-con}
For each $d\geq 1$, there exists an $O(d\log d)$-size program $\p$ such that $\size{\smap{\p}{}{(v)}}= d\size{v}^d$.
\end{lemma}
We say that a CorePolyC program $\p$ simulates a Turing machine $M$ if for each binary number $\alpha$, $M(\alpha)=\beta$ if and only if $\p(\widehat{\alpha})=\widehat{\beta}$, where $\widehat{\alpha}$ for example is an efficient encoding of $\alpha$.
Then the following simulation lemma holds.
\begin{lemma}
\label{lem:simulate}
Every polynomial-time Turing machine can be simulated by a CorePolyC program, and such a program can be constructed in polynomial time.
\end{lemma}
\begin{proof}
We start by recalling the definition of Turing machines.
A \textit{Turing machine} $M$ is a $6$-tuple $(Q, \{\texttt{0,1}\}, \Lambda, \delta, q_0, q_{halt})$, where
\begin{enumerate}
    \item $Q$ is a finite set of states,
    \item $\Lambda=\{\texttt{0,1, $\square$}\}$ where \texttt{$\square$} is the blank symbol,
    \item $q_0,q_{halt}\in Q$ are the start state and the halt state, respectively.
    \item $\delta:Q\times\Lambda\to Q\times\Lambda\times\{\texttt{L,R}\}$ is the transition function of $M$.
\end{enumerate}

The Turing machine $M$ has a single, one-way infinite tape and behaves in a standard manner.
Its configuration is represented by $uqw$, where $q$ is the current state, $uw\in\Lambda^{*}$ is the current tape content, and the head points at the first symbol of $w$.
When $M$ halts, the head will return to the leftmost position, and the content on the tape represents the computation result.

Given $M=(Q, \{\texttt{0,1}\}, \Lambda, \delta, q_0, q_{halt})$ whose running time is given by a polynomial $T(n)$.
There is a constant $d$ such that $T(n)\le dn^d$. We construct a program $\p$  with input variable \texttt{x} as follows:

\paragraph{Configuration.} We use an $\typeInt$ variable \texttt{y} and the input variable \texttt{x} to present the contents of the tape in ternary representation, where the symbols \texttt{0}, \texttt{1}, and \texttt{$\square$} are represented by the digits 0, 1, and 2, respectively.
The value of \texttt{y} represents the contents to the left of the head, while the value of \texttt{x} represents the contents to the right of the head.
For convenience, the digits of \texttt{x} are in reverse order so that the least significant digit of \texttt{x} represents the content under the tape head.
Since the most significant digit cannot be $0$, we add an extra $2$ at the beginning of \texttt{x} and \texttt{y}.
Additionally, we declare an $\typeInt$ variable \texttt{q} to simulate the state of $M$.
W.l.o.g., we assume that $Q=\{q_0,q_1,q_2,\cdots,q_{\size{Q}-1}\}$, where $q_1=q_{halt}$.

\paragraph{Input and Initialization.} Assume that the input on the tape are $w\in\{\texttt{0,1}\}^{*}$.
Let $w'=\texttt{2}\circ w^r$, where $w^r$ represents the reverse of $w$.
Interpret $w'$ as a ternary number and convert it to the corresponding decimal number $v$, which will serve as the input to $\p$.
We declare an \texttt{int} variable \texttt{y} and assign it the value $2$.
For example, if the input of $M$ is $w=$\texttt{10010}, then $w'=(201001)_3$ and $v=(514)_{10}$.

\paragraph{Timer.} We utilize the method presented in the proof of Lemma \ref{lemma:poly-con} to generate a value whose length grows to $d\size{v}^d$, and assign it to an iterable variable \texttt{cnt}.
Then we design a statement of the form \texttt{\textbf{for}(i<\textbf{size}(cnt))\{ $\dots$ \}} such that every Turing machine operation is simulated in an iteration of the loop.

\paragraph{Simulation.} Inside the loop body, we simulate $\delta$ by a sequence of conditional statements.
$M$ will only modify its state once in each step.
Therefore, we declare a boolean variable \texttt{flag} at the beginning of the loop body to indicate whether \texttt{q} has been updated.
For a given constant $m$, we introduce the syntax sugar \texttt{$m$*a}, which means \texttt{a+$\dots$+a} ($m$ occurrences of \texttt{a}).

\begin{enumerate}
    \item \textit{Read}. 
    For any transition $\delta(q_i, \alpha)$, where $i\ne1$, we use the following conditions to simulate the tape’s reading and state recognition.
    If they match, the statement inside \texttt{\{\}} simulates the instructions to be executed.
    $$\texttt{\textbf{if}(!flag \&\& q==$i$ \&\& x\%3==$\alpha$) \{$\dots$\}}.$$
    \item \textit{Write}. If $M$ needs to change the content to $\beta\in\Lambda$ and update the state to $q_j$, it can be done by executing
    $$\texttt{x=x-$\alpha$+$\beta'$; q=}j\texttt{; flag=\textbf{true};},$$
    where $\beta'$ is corresponding to $\beta$, i.e., if $\beta\in\{\texttt{0},\texttt{1}\}$, then $\beta'=\beta$; otherwise $\beta'=\texttt{2}$.
    Set \texttt{flag} to $\t$ to skip the remaining conditional statements.
    \item \textit{Leftward Movement}. If the head of $M$ needs to move left, the last digit of \texttt{y} is appended to the end of \texttt{x}.
    However, it is important to note that if the tape head is already at the leftmost position, no movement is performed. This process can be simulated as follows:
    $$\texttt{\textbf{if}(y>2) \{x=3*x+y\%3; y=y/3;\} \textbf{else} \{\}}$$
    In base three, the value of \texttt{y} always begins with $2$, representing the leftmost end of the tape (\texttt{$\square$}).
    \item \textit{Rightward Movement}. The rightward movement is simulated by appending the last digit of \texttt{x} to the end of \texttt{y}.
    If the head is in the rightmost position, it can still move right.
    In this case, the head will point to a blank symbol \texttt{$\square$}, so we should append digit \texttt{2} to the end of \texttt{y}.
    $$\texttt{y=3*y+x\%3; x=x/3; \textbf{if}(x<=2) \{x=3*x+2\} \textbf{else} \{\}}$$
    \item \textit{Halt and Output.} 
    We do not add any statements for $q_1$. 
    Once the program enters this state, it performs no operations until the loop terminates.
    Add \texttt{\textbf{return} x;} at the end of $\p$ to return the encoded output.
\end{enumerate}
For example, assuming $\delta(q_0,\texttt{0})=(q_1,\texttt{1},\mathtt{L})$ and $\delta(q_0,\texttt{1})=(q_3,\texttt{0},\mathtt{R})$, the corresponding simulation statements are shown in Figure \ref{fig:simulation}.
Since $\delta$ is defined on $Q\times\Lambda$, the number of conditional statements is at most $3\size{Q}$. 
Hence, the entire simulation process can be constructed in polynomial time.
\begin{figure}[t]
	\begin{minipage}{0.48\linewidth}
		\centering
\begin{lstlisting}[xleftmargin=1em,xrightmargin=1em]
for(i<size(cnt)){
  bool flag;
  // ...
  if(!flag&&q==0&&x%3==0){
    x=x-0+1;
    q=1;
    flag=true;
    if(y>2){
        x=3*x+y%3;
        y=y/3;
    } else {}
  } else {}
  // other transitions...
}
\end{lstlisting}
	\end{minipage}
	\hspace{10mm}
	\begin{minipage}{0.48\linewidth}
		\centering
\begin{lstlisting}[xleftmargin=1em,xrightmargin=1em]
for(i<size(cnt)){
  bool flag;
  // ...
  if(!flag&&q==0&&x%3==1){
    x=x-1+0;
    q=3;
    flag=true;
    if(x>2){
        y=3*y+x%3;
        x=x/3;
    } else y=3*y+2;
  } else {}
  // other transitions...
}
\end{lstlisting}
	\end{minipage}
	\caption{Examples of the simulation process.
    The left program simulates the transition $\delta(q_0,\texttt{0})=(q_1,\texttt{1},\mathtt{L})$, 
 , while the right one simulates the transition $\delta(q_0,\texttt{1})=(q_3,\texttt{0},\mathtt{R})$.
    }
 \label{fig:simulation}
\end{figure}

The program $\p$ fully simulates the computation of $M$. Since the tape head of $M$ eventually returns to the leftmost position, the reversed ternary representation of the output without the last digit corresponds to the output of $M$.
\qed
\end{proof}

The completeness result is now immediate.
\begin{theorem}[Completeness]
\label{thm:completeness}
$\classFP\subseteq\classPC$.
\end{theorem}

\section{PolyC}\label{sec:polyc}

As noted in Section~\ref{sec:main_theorem}, the set of CorePolyC computable functions corresponds precisely to $\classFP$.
Important verification properties such as efficient termination are intrinsic to CorePolyC programs, eliminating the need for additional verification. 
Once the theoretical foundation has been laid down, we look for a greater level of programming convenience. 
In this section, we extend CorePolyC with additional syntactic sugars while maintaining the expressive power of CorePolyC.
The extension is not done in a simple-minded way since it is easy to break the polynomial computability barrier. 
We shall call the extended language PolyC. 

\subsection{The Language}
\label{subsec:language}
We outline how to embed in PolyC a few familiar programming constructs.
Some programming examples of PolyC are provided in Section~\ref{sec:example_PolyC}.
We emphasize that we view PolyC as an open and extensible language: features can be introduced, and some design choices can also be modified, as long as such changes do not violate the core design goal: to faithfully characterize $\classFP$.

\subsubsection{Data Types.}
For convenience, we introduce the new compound data types \texttt{array} and \texttt{string}.
The elements in a multi-dimensional array are initialized to $0$ or $\f$.
The size of an array is defined as the maximum length among all elements.
As for strings, there is an iterable version of the string type, i.e., the \texttt{istring} type, such that the operator $\opSize$ is applicable to \texttt{istring} variables.
In addition, the parameters of the main function are not restricted to be only \texttt{int}.

\subsubsection{Functions.}
We are free to use those statements and expressions that can be translated into CorePolyC, such as \texttt{+=}, \texttt{++}, etc.
We can also introduce function declarations and function applications. 
The syntax is as follows.
\begin{align*}
    \texttt{Statement\quad\, s}::=&\texttt{ ... | $\t$ $f$($\tt_1$ $x_1$, $\dots$ $\tt_m$ $x_m$)\{ $\widetilde{\s}$ \textbf{return} $\e$;\}}\\
    \texttt{Expression e}::=&\texttt{ ... | $f$($\widetilde{\e}$)}
\end{align*}
The type of a function is an arrow type $\widetilde{\tt} \to \tt$.
If we do not care about the output of a function, then it can output any value, or we can use the syntax sugar \texttt{void} to declare the return type.
A function value is actually a closure $(\Sigma, \widetilde{x}, \widetilde{\s}, \e)$.
The typing rules and the operational semantics of functions are defined as follows.
{
\small
\[\trule{\typing{\Gamma[\;\widetilde{x}\mapsto\widetilde{\tt}\;]}{\t}{\widetilde{\s}}{\Gamma'}\quad \typing{\Gamma'}{\t}{\e}{\tt}\quad f\not\in\dom(\Gamma)}{\typing{\Gamma}{\l}{\texttt{$\tt$ $f$($\widetilde{\tt\;x}$)\{ $\widetilde{\s}$ \textbf{return} $\e$;\}}}{\Gamma[f\mapsto (\widetilde{\tt}\to \tt)]}}{Fun}\]

\[\trule{\typing{\Gamma}{\l}{\widetilde{\e}}{\widetilde{\tt_1}}\quad\widetilde{\tt_1}\preccurlyeq\widetilde{\tt_2}\quad \Gamma(f)=\widetilde{\tt_2}\to \tt}{\typing{\Gamma}{\l}{\texttt{$f$($\widetilde{\e}$)}}{\tt}}{App}\]

\[\srule{}{\semantics{\Sigma}{\texttt{$t$ $f$($\widetilde{t\;x}$)\{ $\widetilde{\s}$ \textbf{return} $\e$;\}}}{\Sigma[f\mapsto(\Sigma,\widetilde{x},\widetilde{\s},\e)]}}{Fun}\]

\[\srule{\Sigma(f)=(\Sigma_f,\widetilde{x},\widetilde{\s},\e)\quad\semantics{\Sigma}{\widetilde{\e}}{\widetilde{v}}\quad\semantics{\Sigma_f[\;\widetilde{x}\mapsto\widetilde{v}\;]}{\widetilde{s}}{\Sigma'}\quad\semantics{\Sigma'}{\e}{v'}}{\semantics{\Sigma}{\texttt{$f$($\widetilde{\e}$)}}{v'}}{App}\]~
}

\noindent
The partial order relation $(\T_{+}^{m}, \preccurlyeq)$ is a pairwise extension of $(\Int, \preccurlyeq)$ on $m$-tuples, where $\typeIterInt \preccurlyeq \typeInt$.
This indicates that a parameter specified as a non-iterable type in the definition of a function can also be supplied with an iterable value.

It is clear that recursion is not definable in this version of PolyC.
However, this restriction does not significantly impact the expressive power of functions.
The constraints on the iterable variables are also imposed on the body of a function, with the loop indicators $\l$ set to $\t$.
Otherwise, it would be easy to write a function \texttt{mul(a,b)} similar to Example \ref{example:fast-mul} to simulate general multiplication and then invoke \texttt{x=mul(x,x)} in a loop, leading to an exponential explosion of variable lengths.
Therefore, when defining a function, it is necessary to consider that the function may be used within a loop body, and the typing rules above provide a solution to this issue.

Some form of recursion can be incorporated into PolyC, as long as it does not destroy the polynomial computability. 
For instance, simple tail recursion can be incorporated.
More generally, to prevent divergence, we would need to ensure that the size of function parameters strictly decreases under some well-founded ordering (i.e., a termination metric). 
However, since this must be determined statically, defining recursive functions becomes more complex; for instance, some form of decreasing proof must be provided. 
In the simplest case, if we restrict recursion to only allow predefined monotonically decreasing metrics, then recursion would resemble loops in both structure and expressiveness. 
For simplicity, we have not discussed the issues raised by recursion.

\subsubsection{Operators.}
PolyC is compatible with a variety of atomic operators.
Instead of specifying all the operators of PolyC, we would like to clarify what kinds of operators can be included in the operator set.
It is easy to verify that any polynomial-time computable operator satisfying the bound given in Lemma~\ref{lem:op-size} is harmless to PolyC, which falls into the category of \textit{positive operators}, as defined in \cite{hainry2023general}.
However, this condition can be relaxed.
The principle of adding an operator $\op$ to PolyC is to make sure that it is polynomial-time computable and renders true the following predicate.
\begin{center}
    If $\typing{\Gamma}{\f}{\op(\widetilde{x})}{\tt}$ and $\semantics{\Sigma}{\op(\widetilde{x})}{v}$, then $\size{v}\dot{-}\size{\Sigma}=O(\mathrm{poly}(\size{\restr{\Sigma}{\Gamma}}))$.
\end{center}
We can now explore why some operations are unsuitable for PolyC and how such operations can be restricted to meet the aforementioned conditions.
For example, since $\size{s_1 \circ s_2} = \size{s_1} + \size{s_2}$, the general string concatenation must be prohibited, unless we restrict its type to $\mathtt{string}\times\mathtt{istring}\to\mathtt{string}$.
Similarly, scalar multiplication has the type $\typeIterInt \times \typeInt \to \typeInt$, and is therefore safe. 
We already introduced this syntactic sugar earlier in the proof of Lemma~\ref{lem:simulate}.
The sensitivity highlights the fact that we must be very careful when extending PolyC.

\subsection{Coincidence with $\classFP$}

We shall use the subscript ``+'' to annotate the PolyC counterparts of what is defined for CorePolyC.
For clarity, we only consider the PolyC computable functions of type $\Z^m \to \Z$ and denote the corresponding set by $\classPC_{+}$.
Clearly, by Theorem \ref{thm:completeness}, one has $\classFP \subseteq \classPC \subseteq \classPC_{+}$.
We only need to show that $\classPC_{+}\subseteq\classFP$, whose intuition resembles the roadmap shown in Fig.~\ref{fig:dependency_graph}.
Note that the proof of Branch I relies solely on the structure of the program and is unaffected by newly introduced data types or syntactic sugar.
We can derive the following lemma, whose proof elaborates on the proof of Lemma~\ref{lem:class}.

\begin{lemma}\label{lem:polyc-class}
$\classPC_{+}\subseteq \classFP$ if and only if $\,\forall f\,{\in}\,\classPC_{+}.\;|f(\widetilde{v})| =O(\polyn)$.
\end{lemma}

For Branch II, the introduction of functions and function calls in PolyC enhances the growth rate of variable lengths in expressions.
Since functions may have parameters and can be invoked in a nested fashion, the size of return values can grow polynomially.
Lemma \ref{lem:polyc-expr-size} can be seen as an analog of Lemma~\ref{lem:expr-size}, with the constant $k$ replaced by a polynomial overhead $\mathrm{poly}(\size{\restr{\Sigma}{\Gamma}})$.
\begin{lemma}
\label{lem:polyc-expr-size}
Given a PolyC expression $\e$ and a consistent pair $(\Gamma,\Sigma)$, if there exist $v$ and $\tt$ such that $\semantics{\Sigma}{\e}{v}$ and $\typing{\Gamma}{\f}{\e}{\tt}$, then
\begin{enumerate}
    \item $\size{v}=\size{\Sigma}+O(\size{\restr{\Sigma}{\Gamma}}^{d})$ where $d$ is a constant.
    \item if $\tt\in\{\typeIterInt, \mathtt{{istring}}\}$, then $\size{v}=O(\size{\restr{\Sigma}{\Gamma}}^{d})$.
\end{enumerate}
\end{lemma}

However, it should be clear that there is no essential difference between a constant and a polynomial $\mathrm{poly}(\size{\restr{\Sigma}{\Gamma}})$.
Hence, Lemma \ref{lem:stmt-invariant} remains applicable for PolyC.
See Appendix~\ref{app:proof_polyc_size} for the formal proofs.
Moreover, it follows that Proposition~\ref{prop:prog-size} also holds for PolyC, which illustrates that PolyC can only construct outputs of polynomial size. 
By Lemma~\ref{lem:polyc-class}, one has that $\classPC_{+} \subseteq \classFP$.
We can now conclude the expressive power of PolyC.

\begin{theorem}
\label{thm:poly-main}
$\classPC_{+}=\classFP$.
\end{theorem}

Using arrays, the simulation of Turing machines is more straightforward.
However, the completeness proof for CorePolyC does have the virtue in revealing that the four arithmetic operations (\texttt{+}, \texttt{-}, \texttt{\%}, and \texttt{/}), and the \texttt{size} operator, in combination with the loop and the conditional statements, are already sufficient to simulate every polynomial-time Turing machine.

\subsection{Implementation and Examples}\label{sec:example_PolyC}
To demonstrate the practical feasibility of PolyC, we focus on its implementation and present several non-trivial examples in this part.

\subsubsection{Implementation.}
We have implemented an interpreter for PolyC, making use of the tool ANTLR (ANother Tool for Language Recognition)~\cite{parr1995antlr} to generate a syntax tree from the PolyC grammar.
Special attention is given to the handling of the iterable variables.
The source code and the implementation details can be accessed through the supplemental materials.
This tool allows us to do problem-solving in a polynomial time programming paradigm, leveraging its advantages in various programming tasks.
It is our hope that the experiments with this implementation of PolyC will guide the future development of the language.

It is interesting to understand the complexity of the interpreter.
Here again, we see the advantage of PolyC.
In most models for $\classFP$, the termination property is not available, rendering meaningless the complexity issue of any interpreter.
Formally let $\I_{m}:\P_{+}\times{\V_{+}}^{m}\to\V_{+}$ denote the abstract interpreter function for programs of arity $m$, i,e,
$\I_m(\p,\widetilde{v})=\smap{\p}{}(\widetilde{v})$.
When considering only the decision version of this problem, i.e., $\{(\p,\overline{v},i): \p(\overline{v})=i\}$, there is already an intrinsic difficulty in interpreting such programs.

\begin{proposition}
\label{prop:poly-encode}
$\I_0$ is $2$-$\mathbf{EXP}$-complete.
\end{proposition}

Note that the upper bound does not increase as $m$ grows, so we easily obtain the following corollary.
\begin{corollary}
For all $m\ge0$, $\I_m$ is $2$-$\mathbf{EXP}$-complete.
\end{corollary}

Intuitively, the complexity of the interpreter characterizes the simplicity of PolyC, since the more concise the language, the less efficient the interpreter tends to be.
Another way to understand this is that Proposition~\ref{prop:poly-encode} concerns the fact that, for any PolyC program, the length of its execution paths is at most double-exponential; whereas Theorem~\ref{thm:poly-main} characterizes the case where, for a fixed program, the length of its execution paths is guaranteed to be polynomial. 
These two should not be conflated. 
Moreover, in general-purpose programming languages, both scenarios are unbounded, which is one of the reasons why PolyC can simplify the program verification process.

From the theoretical perspective, there remains an open question: for any language characterizing $\classFP$, what is the lower bound for the complexity of the interpreter of the language?
By diagonalization, the complexity of any such interpreter is strictly beyond $\mathbf{P}$.
It is extremely unlikely that such an interpreter is inside $\mathbf{NP}$ or $\mathbf{PSPACE}$ since that would separate $\mathbf{NP}$ or $\mathbf{PSPACE}$ from $\mathbf{P}$.

\subsubsection{More Examples.} 
Next, we shall present several PolyC program examples that make use of the newly added features.

\begin{example}[Knapsack Problem]
\label{exp:knapsack}
In the Knapsack problem, given $m$ items with positive integer weights $w_1, \ldots, w_m$ and positive integer values $c_1, \ldots, c_m$, it is asked to find a subset of items such that the sum of their weights does not exceed a given positive integer $W$, and the total value is maximized.
The most straightforward way to compute the value is by using the dynamic programming algorithm in $O(mW)$-time.
Note that a pseudo-polynomial $O(mW)$ of input length, which rules out any implementation in PolyC.
But if we use the unary representation of $W$, the classic dynamic programming algorithm for \texttt{knapsack}  can be implemented in PolyC.
\begin{lstlisting}
int max(int x,int y){
    if(x>y) return x;
    return y;
}

int knapsack(array<int> w, array<int> v,iint W,iint n){
    array<array<int>> dp=array(size(n));
    for(i<size(n)) dp[i]=array(size(W)+1);
    for(i<size(n)){
        for(j<size(W+W)){
            if(i==0){
                if(j>=w[0]) {
                    dp[i][j]=v[0];
                } else {
                    dp[i][j]=0;
                }
                continue;
            }
            if(j<w[i]) {
                dp[i][j]=dp[i-1][j];
            } else {
                dp[i][j]=max(dp[i-1][j],dp[i-1][j-w[i]]+v[i]);
            }
        }
    }
    return dp[size(n)-1][size(W)];
}
\end{lstlisting}
Note that the parameters \texttt{W} and \texttt{n} in the function \texttt{knapsack} appear in the size function.
Both must be of $\typeIterInt$ type.
For example, when $w = [1,2,2,3,1]$, $v = [1,2,3,4,5]$, $W = 5$, and $n = 5$, one may invoke \texttt{knapsack(w, v, 0b11111, 0b11111)} to get the answer.

\end{example}

\begin{example}[Graph Reachability Problem]
\label{exp:path}
The \texttt{PATH} problem asks whether there exists a path from a given source node $s$ to a target node $t$ in a directed graph $G=(V,E)$.
The problem can be solved by an $O(\size{V}+\size{E})$-time Depth-First Search algorithm (DFS).
Assuming $\size{V} = m$, using an adjacency matrix to record information about a directed graph and encoding it as a string \texttt{adjacent}, we get a boolean representation of the answer by invoking \texttt{path(m, s, t, adjacent)} as follows.
\begin{lstlisting}
bool path(int m,int s, int t,istring adjacent){
    array<bool> visited=array(size(adjacent));
    for(i<size(adjacent)) visited[i]=false;
    array<int> stack=array(size(adjacent));
    stack[0]=s;
    int top=1;
    int current;
    for(i<size(adjacent)){
        if(top!=0){
            top-=1;
            current=stack[top];
            visited[current]=true;
            int row=0;
            for(j<size(adjacent)){
                if(j/m==current){
                    if(adjacent[j]=="1"&&visited[j%m]==false){
                        stack[top]=j%m;
                        top+=1;
                    }
                }
            }
        } else {
            break;
        }
    }
    return visited[t];
}
\end{lstlisting}
We remark that designing a correct PolyC program for a problem is essentially a proof that the problem is in $\classFP$.
Therefore, we have proved that $\mathbf{NL}\subseteq\mathbf{P}$ by this program, since the \texttt{PATH} problem is $\mathbf{NL}$-complete.

\end{example}

The interesting aspect in Example \ref{exp:path} lies not only in representing the adjacency matrix of the graph using strings but also in using a loop with the assistance of a stack to replace the original recursive algorithm of DFS.
In fact, it is possible to write a library in PolyC, allowing us to freely use stacks.
This demonstrates that even though PolyC does not admit recursively defined functions, it is always possible to convert a polynomial recursive algorithm into an equivalent iterative form.
Below, we consider a slightly more complex yet very common problem in real life, the sorting problem.

\begin{example}[Sorting Problem]
\label{exp:sorting}
The input for sorting algorithms is an array $A$, and the goal is to rearrange the elements' positions so that they are sorted in non-decreasing or non-increasing order.
There are various sorting algorithms, and typically, their average time complexity is $\Theta(n\log n)$, where $n$ is the number of elements in the array.
Although the size of individual elements is not explicitly considered, overall it is still a polynomial-time algorithm and can be easily addressed using PolyC.
\begin{lstlisting}
void merge(array<int> arr,int left,int mid,
        int right,array<int> tmp,iint m){
    int index=left,p1=left,p2=mid;
    for(i<size(m)){
        if(p1<mid&&p2<right){
            if(arr[p1]<arr[p2]){
                tmp[index]=arr[p1];index++;p1++;
            } else {
                tmp[index]=arr[p2];index++;p2++;
            }
            continue;
        }
        if(p1<mid){
            tmp[index]=arr[p1];index++;p1++;
            continue;
        }
        if(p2<right) {
            tmp[index]=arr[p2];index++;p2++;
            continue;
        }
        break;
    }
}

void sort(array<int> arr,iint m){
    array<int> tmp=array(size(m));
    int gap=1;
    for(i<size(m)){
        int left=0,mid=min(left+gap,size(m));
        int right=min(mid+gap,size(m));
        for(j<size(m)){
            if(left<size(m)){
                merge(arr,left,mid,right,tmp,m);
                left+=gap+gap;
                mid=min(left+gap,size(m));
                right=min(mid+gap,size(m));
            } else {
                break;
            }
        }
        for(j<size(m)) arr[j]=tmp[j];
        gap=gap+gap;
        if(gap>size(m)) break;
    }
}
\end{lstlisting}
Since the loop statement breaks when the array is sorted, the complexity remains $O(n\log n)$.
The function \texttt{sort} requires an additional parameter specifying the length of the array in unary.
For instance, if there is an array \texttt{arr} of length $5$, we can invoke \texttt{sort(arr, 0b11111)} to sort it.

\end{example}

\subsection{Application}\label{subsec:application}
In this part, we explore the potential applications of PolyC within the field of formal methods.
Indeed, the impacts of PolyC are at least two-fold.

\subsubsection{Verification of PolyC Programs.}
Most verification problems of PolyC programs have an exact algebraic characterization.
Since all well-typed PolyC programs are well-structured, the loops in PolyC can be unrolled such that any nested structures can ultimately be transformed into a form with at most one loop. 
Formally, we say that a program is \emph{simple} if its main function consists of three parts: i) declaration of variables, ii) a loop \texttt{\textbf{for}(i<\textbf{size}(y)) $\s$}, where $\s$ does not contain any loops or function calls, and iii) a return statement.
Proposition~\ref{prop:unfold} demonstrates that any program can be reduced to a simple one.

\begin{proposition}
    \label{prop:unfold}
    For any $m$-ary PolyC program $\p$ with inputs $\widetilde{x}$, there exists a bound $t_0(n)=n^{\size{\p}^{O(\size{\p})}}$ and a $O(\mathrm{poly}(\size{\p}))$-size $(m+1)$-ary simple program $\p_0$ with inputs $(\widetilde{x},y)$ such that
    \begin{enumerate}
        \item the loop bound in the main function of $\p_0$ is $\opSize(y)$, and
        \item for any inputs $\widetilde{v}\in\V^m$ and $\size{t}\ge t_0(\size{\widetilde{v}})$, $\smap{\p_0}{}(\widetilde{v},t)=\smap{\p}{}(\widetilde{v})$.
    \end{enumerate}
\end{proposition}
The idea behind this proposition is simple: since any program implements a function in $\classFP$, they have the normal form as shown in the proof of Lemma~\ref{lem:simulate}.
A constructive proof is given in Appendix~\ref{app:proof_unfold}.

By Proposition~\ref{prop:unfold}, for any program $\p$, we can construct efficiently the corresponding $\p_0$ and $t_0$, where $t_0$ is bounded by a doubly-exponential function of the program size.
Assume that there are $k$ variables declared in the first part of $\p_0$, including the inputs.
Then the effect of the loop body can be viewed as a function $\mathsf{loop}:\V^k\to \V^k$, and the expression in the return statement can be viewed as $\mathsf{ret}:\V^k\to\Z$.
The input stage can be viewed as a function $\mathsf{init}:\V^m \to \V^k$, where any newly declared variables are initialized to their default values.
Finally, the effect of $\p$ can be characterized as $\smap{\p}{}=\mathsf{ret}\circ \mathsf{loop}^{[t_0]}\circ \mathsf{init}$, where the superscript $[t_0]$ denotes the $t_0$-fold composition of the loop body's effect. 
Some semantic properties, for example, the equivalence of two programs $\p_1$ and $\p_2$ can be precisely reduced to a boolean statement of the form
\begin{equation}
    \varphi_{\mathsf{equiv}}(\p_1,\p_2):=\forall \widetilde{v}.\;(\mathsf{ret_1}\circ \mathsf{loop_1}^{[t_1]}\circ \mathsf{init_1}) (\widetilde{v})=(\mathsf{ret_1}\circ \mathsf{loop_2}^{[t_2]}\circ \mathsf{init_2}) (\widetilde{v}).
\end{equation}
Note that a general Turing-complete language cannot provide such algebraic characterizations.
Moreover, in daily life, we only consider bounded inputs.
Then such formulas can be viewed as instances of co-$\mathbf{NP}$ problem $\mathbf{TAUTOLOGY}$. 
Furthermore, by considering over-approximations or restricting to subclasses of programs, more lightweight analysis methods may be achieved.

\subsubsection{Impact on Other Programming Languages.}
Our type system can also serve as a complexity analysis tool for other programming languages.
A straightforward approach would involve the following steps, as shown in Algorithm~\ref{alg}.

\begin{algorithm}[ht]
\caption{A simple complexity analysis procedure}
\begin{algorithmic}[1]\label{alg}
    \STATE Assign iterable types to all variables;
    \REPEAT
        \STATE Perform type checking;
        \IF{type checking succeeds}
            \STATE \textbf{output} ``poly''; \RETURN
        \ELSE
            \STATE Update some types to non-iterable types accordingly;
        \ENDIF
    \UNTIL{all involved types cannot be adjusted}
    \STATE \textbf{output} ``unknown'';
\end{algorithmic}
\end{algorithm}

This algorithm is sound but not complete.
Consider the left program in Fig.~\ref{fig:mul} as an example.
In the first step, \texttt{y} is assigned the type $\typeIterInt$.
Since \texttt{y} is modified within the loop body, its type must eventually be updated to $\typeInt$.
However, because \texttt{y} also serves as the operand of the $\opSize$ operator, it must retain the type $\typeIterInt$.
As a result, the program fails to pass the polynomial check.
Nevertheless, this does not imply that its running time is non-polynomial.
A key challenge in enhancing this method lies in inferring whether loop bounds are dominated by the bit-size of an iterable expression. 
For instance, \texttt{y} can easily be replaced by a newly declared iterable variable \texttt{z}.
If such substitutions can be effectively detected, the class of programs for which the tool can soundly determine the ``poly'' status would be significantly extended.

\section{Remark}\label{sec:conclusion}
What have we achieved?
We have designed a simple imperative programming language, provided a rigorous definition of its syntax, semantics, and type system, and showed the consistency of the type system.
We have proved that PolyC captures precisely the class $\classFP$ of the polynomial time computable functions, and thus have laid down a firm foundation for a programming methodology.
The design of {PolyC} has successfully addressed the problems that appeared in the previous models and languages. 
It is discussed in the supplementary materials how to introduce more features in {PolyC}.

What can {PolyC} offer?
It is clear that PolyC forces a particular programming methodology on programmers.
Designing a correct program for a problem is itself a proof that the problem is polynomial time computable (See Example~\ref{exp:path}).
Programming in PolyC calls for more low-level algorithmic thinking.
Programmers may spend more time in the design phase, trying to understand the problem better before programming.
The language's built-in complexity boundaries create a natural simplification to theorem provers and model checkers: by eliminating infinite state space inherent to general-purpose languages, PolyC reduces verification overhead and may even enable lightweight verification strategies.

What more can we do with {PolyC}?
Some software systems implement algorithms for $\mathbf{NP}$-complete problems.
In reality, their performance is just fine.
The reason is that among all the input parameters of any one such system, the algorithm depends only on some input exponentially, and the size of the input tends to be small in applications.
The problems solved by these systems are studied in the theory of parameterized complexity~\cite{downey2012parameterized}.
There is more than one way to implement such systems in {PolyC}.
Either a constant upper bound is placed on an input parameter, or the input is coded up in unary.
The dynamic algorithm of the knapsack problem, for example, can be implemented in {PolyC} using the unary encoding of the parameters (Example~\ref{exp:knapsack}).
Thus {PolyC} is far more useful than what it appears at first sight.
Further development of the type system is also a promising future work.
As discussed in Section~\ref{subsec:application}, we can slim down programs designed in other languages to a PolyC-like form for the purpose of feasibility, where the treatment of the loop constructs is crucial.

Randomization is becoming more and more useful in algorithm design.
It is shown in~\cite{10.1145/380752.380813} and the follow-up works that the simplex method has randomized algorithms whose average running time is polynomial.
It means that we may design with a randomized PolyC a program for the linear programming problem that has a great success rate.
In the setting of randomized computation, a polynomial-time algorithm is all we need. 
PolyC can, of course, be extended in different manners in different application scenarios.

\subsubsection{Acknowledgements.}
This work is supported by the National Natural Science Foundation of China (No. 62072299, 62020106005) and the Science and Technology Commission of Shanghai Municipality (No. 24BC3200500).
We are particularly grateful to the anonymous reviewers and Xin Yu for their valuable comments, which have greatly improved this work.

\bibliographystyle{splncs04}
\bibliography{contents/reference}

\appendix

\section{Missing Proofs of Section \ref{sec:CorePolyC}}
\subsection{Proof of Lemma~\ref{lem:map}}
\begin{proof}
The proof is by structural induction. 
\begin{enumerate}
    \item By definition $\C=\{\True,\False\}\cup\D^{+}$.
By the definitions of $\smap{\cdot}{}$ and $\tmap{\cdot}{}$, the followings are valid:
    \begin{itemize}
        \item $\smap{\c}{}\in\B$ and $\tmap{\c}{}=\typeBool\in\Bool$ whenever $\c\in\{\True,\False\}$;
        \item $\smap{\c}{}\in\Z$ and $\tmap{\c}{}=\typeIterInt\in\Int$ whenever $\c\in \D^{+}$.
    \end{itemize}
    Therefore $\smap{\c}{}$ and $\tmap{\c}{}$ are always consistent.
    \item $\T=\Bool\cup\Int$, and
    \begin{itemize}
        \item $\smap{t}{}=\f\in\B$ whenever $t\in\Bool$;
        \item $\smap{t}{}=0\in\Z$ whenever $t\in\Int$.
    \end{itemize}
Therefore $\smap{t}{}$ and $t$ are consistent.
    \item Assume that $\op\in\{\texttt{\&\&},\texttt{||},\texttt{!}\}$.
    Since $\widetilde{t}\in\dom(\tmap{\op}{})$, it can only be the case that $\widetilde{t}=\typeBool^m$, and thus $\tmap{\op}{}(\widetilde{t})=\typeBool$. Furthermore, as $\widetilde{v}$ is consistent with $\widetilde{t}$, it follows that $\widetilde{v}\in\B^m$.
    Regardless of whether $\op$ is \texttt{\&\&}, \texttt{||}, or \texttt{!}, it holds that $\smap{\op}{}(\widetilde{v})\in\B$.
    Consequently, $\smap{\op}{}(\widetilde{v})$ is consistent with $\tmap{\op}{}(\widetilde{t})$.
    Similar reasoning can be applied in other cases.
    \qed
\end{enumerate}
\end{proof}
\subsection{Proof of Lemma~\ref{lem:preservation}}

\begin{proof}
The proof is by structural induction on the typing derivation trees.
The uniqueness follows from the deterministic nature of the semantics.

As for the first statement, consider the derivation tree of $\typing{\Gamma}{\l}{\e}{t}$.
Case~\tref{Const} and \tref{Op} are proved by applying Lemma~\ref{lem:map} (1) and (3).
Case \tref{Var} follows from the consistency of $\Sigma$ and $\Gamma$.
Case~\tref{Paren} is evident.

As for the second statement, $\Gamma'\subseteq\Gamma$ and the existence of $\Sigma'$ are trivial.
We only prove the consistency of $\Sigma'$ and $\Gamma'$ by induction.
\begin{itemize}
    \item Case~\tref{Asgmt}.
    Since $x\in\dom(\Gamma)$, we have $\dom(\Sigma')=\dom(\Sigma)\cup\{x\}\subseteq\dom(\Gamma)$.
    Therefore, one only needs to prove that $v$ and $\Gamma(x)$ are consistent, which is indeed true since $v$ and $t$ are consistent by Lemma~\ref{lem:preservation} (1) and $t\comparable\Gamma(x)$.
    \item Case~\tref{Decl}.
    Since $\Gamma'=\Gamma[x\mapsto t]$ and $\Sigma'=\Sigma[x\mapsto\smap{t}{}]$, we can complete our proof by applying Lemma \ref{lem:map} (2).
    \item Case~\tref{Loop} where $\;\s\equiv\texttt{\textbf{while}(}x \texttt{<} \e \texttt{)} \s'$.
    $$\trule
    {\typing{\Gamma}{\l}{\e}{t\equiv\typeIterInt}\enspace x\not\in\dom(\Gamma)\enspace\typing{\Gamma[x\mapsto{t}]\,}{\t}{\s'}{\Gamma'}}
    {\typing{\Gamma}{\l}{\s}{\Gamma}}
    {Loop}$$
    \[\srule{\semantics{\Sigma}{\e}{i}\enspace\Sigma_0:=\Sigma\enspace\semantics{\Sigma_{j}[x\mapsto j]}{\s'}{\Sigma_{j+1}},\;\text{for}\;\;0\le j< i.}{\semantics{\Sigma}{\,\s}{\Sigma_{i}}}{Loop}\]
    Since $\typing{\Gamma}{\l}{\e}{t} \in \Int$, we obtain $\semantics{\Sigma}{\e}{i \in \Z}$ by the induction hypothesis.
    If $i\le 0$, the case is trivial.
    Otherwise, since $\typing{\Gamma[x\mapsto t]\,}{\t}{\s'}{\Gamma'}$ and $\Sigma_0[x\mapsto 0]$ is consistent with $\Gamma[x\mapsto t]$, one has that $\semantics{\Sigma_0[x\mapsto 0]}{\s'}{\Sigma_1}$ and $\Sigma_1$ is consistent with $\Gamma'$.
    We have proved that $\Gamma[x\mapsto{t}]\subseteq\Gamma'$. Consequently, $\forall\; x\in\dom(\Gamma)$, $\Gamma(x)=\Gamma'(x)$ and $\Gamma'(x)$ is consistent with $\Sigma_1(x)$. It follows that $\Sigma_1$ is also consistent with $\Gamma$. By repeating this process, one can ultimately prove that $\semantics{\Sigma_{i-1}}{\s'}{\Sigma_i}$ and $\Sigma_i$ is consistent with $\Gamma$. The result follows immediately by applying rule \sref{Loop}.
    \item Case~\tref{Cond}.
    By the induction hypothesis, $\Sigma_1$ is consistent with $\Gamma_1$.
    Since $\Gamma\subseteq\Gamma$, $\Sigma_1$ is consistent with $\Gamma$.
    Similarly, we have $\Sigma_2$ is consistent with $\Gamma$.
    Note that $\Sigma'$ must be either $\Sigma_1$ or $\Sigma_2$ and $\Gamma'=\Gamma$.
    It follows that $\Sigma'$ is consistent with $\Gamma'$.
    \item Case \tref{Seq} and Case \tref{Block} can all be proved by applying the induction hypothesis.\qed
\end{itemize}
\end{proof}

\subsection{Proof of Theorem~\ref{thm:safety}}
\begin{proof}
Assume that $\p\equiv\texttt{\textbf{int main}(\textbf{int} } x_1,\cdots,\texttt{\textbf{int} }x_m \texttt{)\{ }\widetilde{\s}\texttt{ \textbf{return} }\e\texttt{;\}}$.
Let $\widetilde{t}$ denote the parameter types of the main function. By the syntax of CorePolyC, one has that $\widetilde{t}\in\Int^m$.
Therefore, $\emptyset[\widetilde{x}\mapsto\widetilde{i}]$ is consistent with $\emptyset[\widetilde{x}\mapsto\widetilde{t}]$ since $\widetilde{i}\in\Z^m$.
By the well-typedness of program $\p$, there exists $t'$ such that $\typing{\emptyset}{\f}{\p}{t'}$.
According to rule \tref{Prog}, there exist a typing environment $\Gamma'$ and a type $t'\in\Int$ rendering true the following.
\begin{equation}
    \typing{\emptyset[\widetilde{x}\mapsto\widetilde{t}]}{\l}{\widetilde{\s}}{\Gamma'},\quad\typing{\Gamma'}{\l}{\e}{t'}.
\end{equation}
By Lemma \ref{lem:preservation}, we have there exists a unique $\Sigma'$ such that $\semantics{\emptyset[\widetilde{x}\mapsto{\widetilde{i}}]}{\widetilde{\s}}{\Sigma'}$ and $\Sigma'$ is consistent with $\Gamma'$.
Again by the lemma, there exists a unique $v'\in\V$ such that $\semantics{\Sigma'}{\e}{v'}$ and that $v'$ is consistent with $t'$.
Together with $t'\in\Int$, one can prove that $v'\in\Z$.
Finally we conclude from \sref{Prog} that $\semantics{\emptyset[\widetilde{x}\mapsto{\widetilde{i}}]}{\p}{v'}$.
\end{proof}

\subsection{Proof of Proposition~\ref{prop:invariant}}
\begin{proof}
The proof is by case analysis.
\begin{itemize}
    \item In the case of~\tref{Decl}, $\assignable(\t,\tt)$ and $x\not\in\dom(\Gamma)$. By definition of $\assignable$, $\tt\not\equiv\typeIterInt$. Therefore, $\domi{\Gamma'}=\domi{\Gamma[x\mapsto t]}=\domi{\Gamma}$ and $\restr{\Sigma'}{\Gamma'}\sim\restr{\Sigma[x\mapsto \smap{t}{}]}{\Gamma}\sim\restr{\Sigma}{\Gamma}$.
    \item In the case of~\tref{Asgmt}, $\Gamma'=\Gamma$, so $\domi{\Gamma'}=\domi{\Gamma}$ is trivial. We also have that $x\not\in\domi{\Gamma}$ since $\assignable(\t,\Gamma(x))$. Consequently, $\restr{\Sigma'}{\Gamma'}\sim\restr{\Sigma[x\mapsto v]}{\Gamma}\sim\restr{\Sigma}{\Gamma}$.
    \item In the case of~\tref{Seq}, the result follows by a simple induction on $m$.
    \item In the case of~\tref{Block}, \tref{Loop} and \tref{Cond}, we have $\Gamma'=\Gamma$.
    Therefore, $\domi{\Gamma'}=\domi{\Gamma}$ and $\restr{\Sigma'}{\Gamma'}\sim\restr{\Sigma'}{\Gamma}$.
    By the induction hypothesis, one can easily prove that $\restr{\Sigma'}{\Gamma}\sim\restr{\Sigma}{\Gamma}$.\qed
\end{itemize}
\end{proof}

\section{Missing Proofs of Section~\ref{sec:main_theorem}}\label{app-soundness}

\subsection{Cost Semantics of CorePolyC}
\label{app:cost_semantics}

We present the complete cost semantics of CorePolyC as follows.
\begin{figure}[t]
    \centering
    \[
\srule{x\in\dom(\Sigma)}{\cost{\Sigma}{x}{\Sigma(x)}{1}}{Var}\qquad
\srule{\c\in\C}{\cost{\Sigma}{\c}{\smap{\c}{}}{1}}{Const}\qquad
\srule{\cost{\Sigma}{\e}{v}{k}}{\cost{\Sigma}{\texttt{(e)}}{v}{k+1}}{Paren}\]
\[
\srule{\cost{\Sigma}{\widetilde{\e}}{\widetilde{v}}{\widetilde{k}}\quad \op\in\O}{\cost{\Sigma}{\op(\widetilde{\e})}{\smap{\op}{}(\widetilde{v})}{1+\sum_{i=1}^{\arity(\op)}k_i}}{Op}\qquad
\srule{x\in\dom(\Sigma)\quad \cost{\Sigma}{\e}{v}{k}}{\cost{\Sigma}{x\texttt{=}\e\texttt{;}}{\Sigma[x\mapsto{v}]}{k+1}}{Asgmt}
\]
\[
\srule{}{\cost{\Sigma}{\tt\;x\texttt{;}}{\Sigma[x\mapsto{\smap{\tt}{}}]}{1}}{Decl}\qquad
\srule{\cost{\Sigma}{\;\widetilde{\s}\;}{\Sigma'}{k}}{\cost{\Sigma}{\texttt{\{ }\widetilde{\s}\texttt{ \}}}{\Sigma'}{k+1}}{Block}
\]
\[
\srule{\Sigma_0:=\Sigma\qquad\cost{\Sigma_{i-1}}{\s_i}{\Sigma_{i}}{k_i} ,\quad\text{for}\;1\le i\le m.}{\cost{\Sigma}{\;\widetilde{\s}\;}{\Sigma_m}{\sum_{i=1}^{m}k_i}}{Seq}\]
\[
\srule{\cost{\Sigma}{\e}{b}{k}\enspace (\Sigma',k'):=\begin{cases}(\Sigma_1,k_1)\,,&\text{if}\;b\land \cost{\Sigma}{\s_1}{\Sigma_1}{k_1}, \\
(\Sigma_2,k_2)\,,&\text{if}\;\lnot b\land \cost{\Sigma}{\s_2}{\Sigma_2}{k_2}.\end{cases}
}
{\cost{\Sigma}{\texttt{\textbf{if}(}\e\texttt{) } \s_1 \texttt{ \textbf{else} } \s_2}{\Sigma'}{k+k'}}{Cond}
\]
\[
\srule{\cost{\Sigma}{\e}{i}{k}\quad\Sigma_0:=\Sigma\enspace\cost{\Sigma_{j}[x\mapsto j]}{\s}{\Sigma_{j+1}}{k_j},\;\text{for}\;\;0\le j< i.}{\cost{\Sigma}{\texttt{\textbf{for}(}x \texttt{<} \e \texttt{) } \s}{\Sigma_{i}}{k+\sum_{j=0}^{i-1}k_j}}{Loop}
\]\vspace{1mm}
\[
\srule{\cost{\Sigma}{\;\widetilde{\s}\;}{\Sigma'}{k_1}\quad\cost{\Sigma'}{\e}{i}{k_2}}{\cost{\Sigma}{\texttt{\textbf{int} \textbf{main}(}\widetilde{t\;x}\texttt{)\{ }\widetilde{\s}\texttt{ \textbf{return} }\e\texttt{;\}}}{i}{k_1+k_2}}{Prog}
\]
    \caption{The cost semantics of CorePolyC.}
    \label{fig:cost_semantics}
\end{figure}
The following example gives a straightforward illustration of these rules.
\begin{example}[Loop]
\label{exp:loop_cost}
Let $\s:=\;$\texttt{\textbf{for}(i<\textbf{size}(z)) x=x+x;} and $\Sigma = [\texttt{z}\mapsto 2^n,\texttt{x}=1]$.
We have $\cost{\Sigma}{\s}{\Sigma'}{4n+6}$.
{\fontsize{7pt}{6pt}\selectfont
$$
\inference
{
    {
    \inference
    {
        {
        \inference{\texttt{z}\in\dom(\Sigma)}
        {\cost{\Sigma}{\texttt{z}}{2^n}{1}}
        }
    }
    {
        \cost{\Sigma}{\texttt{\textbf{size}(z)}}{n+1}{2}
    }
    }
    \enspace
    \Sigma_0:=\Sigma
    \enspace
    {
    \inference
    {
        \texttt{x}\in\dom(\Sigma_j')
        \quad
        {
        \inference
        {
            {
            \inference
            {
                \texttt{x}\in\dom(\Sigma_j')
            }
            {
                \cost{\Sigma_j'}{\texttt{x}}{2^j}{1}
            }
            }
        }
        {
            \cost{\Sigma_j'}{\texttt{x+x}}{2^{j+1}}{3}
        }
        }
    }
    {
        \cost{\Sigma_j'}{\texttt{x=x+x}}{\Sigma_{j+1}}{4}
    }
    }
    , \text{ for }0\le j\le n.
}
{\cost{\Sigma}{\s}{\Sigma'\equiv\Sigma_{n+1}}{4n+6}}
$$
}
where $\Sigma_j'=\Sigma_j[\texttt{i}\mapsto j]$ and $\Sigma_{j+1}=\Sigma_j'[\texttt{x}\mapsto 2^{j+1}]$ for $0\le j\le n$.
\end{example}

\subsection{Proof of Lemma \ref{lem:time-size}}\label{App_proof_of_prop:time-size}

\begin{proof}
The proof is by contradiction.
For convenience, we define $\ic(w, \Sigma):=t$ such that $\cost{\Sigma}{w}{\cdot}{t}$ for $w\in\E\cup\S$.

\vspace*{1mm}
\noindent
(1) $\Rightarrow$ (2).\quad
Suppose that there exists a program $\p$ such that $\size{\smap{\p}{}}\neq O(T(n))$. By (1), $\p$ has an instruction count of $O(T(n))$.
W.l.o.g., we assume that $\p\equiv \texttt{\textbf{int main}(}\cdots\texttt{)\{ }\widetilde{\s}\;\texttt{ \textbf{return} }\e\texttt{ \}}$ and $\texttt{z},\texttt{o}\not\in\L(\p)$.
A well-typed program $\p'$ can be constructed from $\p$ as follows:
\begin{lstlisting}
int main(`$\cdots$`){
    `$\widetilde{\s}$`
    iint z; z=`$\e$`;
    int o;  o=1;
    for(i<size(z)) o=o+o;
    return o;
}
\end{lstlisting}
By \sref{Prog}, there exists an environment $\Sigma'$ such that $\semantics{\emptyset[\widetilde{x}\mapsto\widetilde{v}]}{\widetilde{\s}}{\Sigma'}$ and $\semantics{\Sigma'}{\e}{\smap{\p}{}(\widetilde{v})}$.
Therefore, by Example \ref{exp:loop_cost},
$$\ic(\texttt{\textbf{for}}\cdots,\Sigma'[\texttt{z}\mapsto \smap{\p}{}(\widetilde{v}),\,\texttt{o}\mapsto1])=4\size{\smap{\p}{}(\widetilde{v})}+2.\vspace{2mm}$$
And $\ic(\p', \widetilde{v})=\ic(\p,\widetilde{v})+\Theta(\size{\smap{\p}{}(\widetilde{v})})$.
Since $\ic(\p,\widetilde{v})=O(T(n)))$ and $\size{\smap{\p}{}}\neq O(T(n))$, we have that $\ic(\p',\widetilde{v})\neq O(T(n))$, i.e., $\p'$ does not have an instruction count of $O(T(n))$, which is a contradiction.

\vspace*{1mm}
\noindent
(2) $\Rightarrow$ (3).\quad Let $\p$ be a program such that $\size{\evalTree{\emptyset[\widetilde{x}\mapsto{\widetilde{v}}]}{\p}}\neq O(T(n))$.
Suppose that \texttt{o}$\not\in\L(\p)$. We construct a new well-typed program $\p'$ from $\p$ as follows:
\begin{itemize}
    \item[a)] Insert the declaration \texttt{\textbf{int} o;} in the beginning of $\p$. For each input $x_i$, where $1 \le i \le m$, append the following statement right after the declaration.
\begin{equation}
\texttt{\textbf{if}(}x_i\texttt{>o)\{o=}x_i\texttt{;\} \textbf{else}\{\} \textbf{if}(-}x_i\texttt{>o)\{o=-}x_i\texttt{;\} \textbf{else}\{\}}.
\end{equation}
    \item[b)] After each assignment in the form of \texttt{x=e;} where $\e$ is typed as an integer, add a new statement
\begin{equation}\label{2023-07-09}
\texttt{\textbf{if}(x>o)\{o=x;\} \textbf{else}\{\} \textbf{if}(-x>o)\{o=-x;\} \textbf{else}\{\}}.
\end{equation}
    Then enclose the assignment statement and~(\ref{2023-07-09})  with $\texttt{\{\}}$.
    \item[c)] Replace return statement \texttt{\textbf{return} e;} by
    \begin{equation*}
    \texttt{\textbf{if}(e>o)\{o=e;\} \textbf{else}\{\} \textbf{if}(-(e)>o)\{o=-(e);\} \textbf{else}\{\} \textbf{return} o;}.
    \end{equation*}
\end{itemize}
We will denote this transformation as $T_1$. For example, the program on the left in Fig.~\ref{fig:T1} is transformed into the program on the right.
\begin{figure}[t]
\begin{minipage}{0.45\linewidth}
\centering
\begin{lstlisting}[xleftmargin=0em,xrightmargin=1.5em]
int main(int x,int y){
    iint z;
    z=y;
    for(i<size(z)) x=x+x;
    return x+y;
}
\end{lstlisting}
\end{minipage}
\hspace{10mm}
\begin{minipage}{0.45\linewidth}
\centering
\begin{lstlisting}[xleftmargin=-1em,xrightmargin=0em]
int main(int x,int y){
    iint z;
    int o;
    {
        z=y;
        if(z>o){o=z;}else{}
        if(-z>o){o=-z;}else{}
    }
    for(i<size(z)){
        x=x+x;
        if(x>o){o=x;}else{}
        if(-x>o){o=-x;}else{}
    }
    if(x+y>o){o=x+y;}else{}
    if(-(x+y)>o){o=-(x+y);}else{}
    return o;
}
\end{lstlisting}
\end{minipage}
\caption{An example of transform $T_1$.}
\label{fig:T1}
\end{figure}
In the second step, we capture the longest value during the program's execution using \texttt{o}. Assume that $\p'\equiv \texttt{\textbf{int} \textbf{main}(}\dots\texttt{)\{ }\widetilde{\s}\;\texttt{ \textbf{return} o;\}}$.
By \sref{Prog} and \sref{Var}, there exists an environment $\Sigma'$ such that $\semantics{\emptyset[\widetilde{x}\mapsto\widetilde{v}]}{\s'}{\Sigma'}$ and $\semantics{\Sigma'}{\texttt{o}}{\Sigma'(\texttt{o})=\smap{\p'}{}{(\widetilde{v})}}$.

\begin{claim}
    $\size{\Sigma'(\texttt{o})}=\size{\evalTree{\emptyset[\widetilde{x}\mapsto{\widetilde{v}}]}{\p'}}\ge\size{\evalTree{\emptyset[\widetilde{x}\mapsto{\widetilde{v}}]}{\p}}$.
\end{claim}

To be more concise, we will provide proof of the claim later.
Consequently, $\size{\smap{\p'}{}{(\widetilde{v})}}=\size{\Sigma'(\texttt{o})}=\size{\evalTree{\emptyset[\;\widetilde{x}\mapsto{\widetilde{v}}\;]}{\p'}}\ge\size{\evalTree{\emptyset[\;\widetilde{x}\mapsto{\widetilde{v}}\;]}{\p}}\neq O(T(n))$, which contradicts to the fact that all CorePolyC programs can only produce outputs of size $O(T(n))$.

\vspace*{1mm}
\noindent
(3) $\Rightarrow$ (1).\quad
Let $\p$ be a program with $\ic(\p,\;\widetilde{v})\neq O(T(n))$.
Suppose \texttt{o}$\not\in\L(\p)$. We construct a new program $\p'$ as follows:
\begin{itemize}
    \item[a)] Insert the declaration \texttt{\textbf{int} o;} in the beginning and initialize \texttt{o} with $1$;
    \item[b)] Add \texttt{o=o+o;} after each declaration statement, assignment statement, if statement, and loop statement in $\p$, and combine the two statements into a block statement; if a statement is an empty block, we also add \texttt{o=o+o;} inside it.
\end{itemize}
This transform is denoted as $T_2$. For example, the program on the left in Fig.~\ref{fig:T2} is transformed into the program on the right.
\begin{figure}[t]
\begin{minipage}{0.45\linewidth}
\centering
\begin{lstlisting}[xleftmargin=1em,xrightmargin=0.5em]
int main(int x,int y){
    iint z;
    z=y;
    for(i<size(z)) x=x+x;
    return x+y;
}
\end{lstlisting}
\end{minipage}
\hspace{10mm}
\begin{minipage}{0.45\linewidth}
\centering
\begin{lstlisting}[xleftmargin=0.5em,xrightmargin=1em]
int main(int x,int y){
    int o;
    o=1;
    {iint z; o=o+o;}
    {z=y;    o=o+o;}
    for(i<size(z)){
        x=x+x;
        o=o+o;
    }
    return x+y;
}
\end{lstlisting}
\end{minipage}
\caption{An example of transform $T_2$.}
\label{fig:T2}
\end{figure}
In the second step, a variable $\texttt{o}$ is introduced whose length increases linearly with $\ic$. Note that the declarations and the assignment statements from the original program are transformed into compound statements.

\begin{claim}
    For any cost semantics rule instance $\cost{\Sigma}{T_2(\s)}{\Sigma'}{k}$ in derivation tree $\evalTree{\emptyset[\widetilde{x}\mapsto\widetilde{v}]}{\p'}$, where $\s$ is a statement in $\p$,  if $\texttt{o}\in\dom(\Sigma)$, then $k \le d \cdot (\size{\Sigma'(\texttt{o})}-\size{\Sigma(\texttt{o})})$; otherwise, $k \le d \cdot \size{\Sigma'(\texttt{o})}$, where $d$ is a constant depending only on $\s$.
\end{claim}

Assume that $\p\equiv \texttt{`\textbf{int main}(}\dots\texttt{)\{ }T_2(\widetilde{\s})\;\texttt{ \textbf{return} e;\}}$.
By \sref{Prog}, there exists an environment $\Sigma'$ such that $\semantics{\emptyset[\widetilde{x}\mapsto\widetilde{v}]}{T_2(\widetilde{\s})}{\Sigma'}$.
This rule instance is surely in the derivation tree $\evalTree{\emptyset[\widetilde{x}\mapsto\widetilde{v}]}{\p'}$.
Applying this claim, we obtain that $d\cdot\size{\Sigma'(\texttt{o})}\ge\ic(T_2(\widetilde{s}),\emptyset[\widetilde{x}\mapsto\widetilde{v}])$ for some $d$.
Since $\ic(\p,\widetilde{v}) \le \ic(\p',\widetilde{v}) \le d\cdot\size{\Sigma'(\texttt{o})} + O(1)$, we have $\size{\Sigma'(\texttt{o})} \neq O(T(n))$, which implies that $\size{\evalTree{\emptyset[\;\widetilde{x}\mapsto\widetilde{v}\;]}{\p'}}\ge\size{\Sigma'}\ge\size{\Sigma'(\texttt{o})}\neq O(T(n))$. But (3) states that $\p'$ can only generate values of size $O(T(n))$.
This is a contradiction.
\qed
\end{proof}

\subsubsection{Additional Proof}\label{App_Claims}
\begin{claim}
    $\size{\Sigma'(\texttt{o})}=\size{\evalTree{\emptyset[\widetilde{x}\mapsto{\widetilde{v}}]}{\p'}}\ge\size{\evalTree{\emptyset[\widetilde{x}\mapsto{\widetilde{v}}]}{\p}}$.
\end{claim}

\begin{proof}
It is sufficient to prove that for all rule instances of the form
$\semantics{\Sigma_1}{\s}{\Sigma_2}$ in $\evalTree{\emptyset[\widetilde{x}\mapsto{\widetilde{v}}]}{\p'}$, $\size{\Sigma_2}=\size{\Sigma_2(\texttt{o})}$ holds, which is done by structural induction.
We only need to consider~\sref{Decl} and \sref{Asgmt} because other cases do not directly alter $\Sigma$, and we can simply apply the induction hypothesis.
\begin{itemize}
    \item Case~\sref{Decl}. Since $\size{\smap{t}{}}=0$ for any $t\in\T$, $\size{\Sigma[x\mapsto\smap{t}{}]}=\max\{\size{\Sigma},0\}=\size{\Sigma}$.
    The result follows by applying the induction hypothesis.
    \item Case~\sref{Asgmt}.
    $\p$'s assignments always appear within a block alongside two conditional statements in new program $\p$.
    We only need to consider these fundamental blocks.
    Assume that a block contains $\widetilde{\s}$, where $\s_1 \equiv \texttt{x=$\e$;}$.
    Without loss of generality, let $\Sigma_0 = \Sigma, \semantics{\Sigma_{i-1}}{\s_i}{\Sigma_i}$ for $1\le i\le 3$. Then:
    \begin{itemize}
        \item if $\Sigma_1(x)>\Sigma_1(\texttt{o})$, then according to~\sref{Cond}, we have $\Sigma_3(\texttt{o})=\Sigma_2(\texttt{o})=\Sigma_1(x)$.
        \item if $-\Sigma_1(x)>\Sigma_1(\texttt{o})$, we also have $\Sigma_3(\texttt{o})=\Sigma_2(x)=\Sigma_1(x)$.
        \item if $\funAbs{\Sigma_1(x)}\le\Sigma_1(\texttt{o})$, then $\Sigma_3(x)=\Sigma_1(x)\le\Sigma_1(\texttt{o})=\Sigma_3(\texttt{o})$.
    \end{itemize}
    Therefore, $\size{\Sigma_3}=\max\{\size{\Sigma},\size{\Sigma_3(x)},\size{\Sigma_3(\texttt{o})}\}=\size{\Sigma_3(\texttt{o})}$.
\end{itemize}
Furthermore, before the return statement, there are two conditional statements.
Similar argument leads to the following conclusion: $\size{\Sigma'(\texttt{o})}=\size{\Sigma'}$.
Since the absolute value of \texttt{o} is always non-decreasing, we obtain that $\size{\Sigma'(\texttt{o})}=\size{\evalTree{\emptyset[\widetilde{x}\mapsto{\widetilde{v}}]}{\p'}}$.
Note that in case~\sref{Asgmt}, $\Sigma_3(x)=\Sigma(x)$ for all $x\not\equiv\texttt{o}$ and only statements related to the variable \texttt{o} are missing in $\p$, we can prove that $\size{\evalTree{\emptyset[\widetilde{x}\mapsto{\widetilde{v}}]}{\p'}}\ge\size{\evalTree{\emptyset[\widetilde{x}\mapsto{\widetilde{v}}]}{\p}}$ by induction.\qed
\end{proof}

\begin{claim}
    For any cost semantics rule instance $\cost{\Sigma}{T_2(\s)}{\Sigma'}{k}$ in derivation tree $\evalTree{\emptyset[\widetilde{x}\mapsto\widetilde{v}]}{\p'}$, where $\s$ is a statement in $\p$,  if $\texttt{o}\in\dom(\Sigma)$, then $k \le d \cdot (\size{\Sigma'(\texttt{o})}-\size{\Sigma(\texttt{o})})$; otherwise, $k \le d \cdot \size{\Sigma'(\texttt{o})}$, where $d$ is a constant depending only on $\s$.
\end{claim}

\begin{proof}
If \texttt{o} $\not\in\dom(\Sigma)$, we can assume that $\size{\Sigma(\texttt{o})} = 0$, thereby rendering the two cases equivalent.
By induction on the structure of $\evalTree{\emptyset[\widetilde{x}\mapsto\widetilde{v}]}{\p'}$:
\begin{itemize}
    \item If $\s$ is a declaration statement or an assignment statement, $T_2(\s)$ is a compound statement. We have the base cases where $\size{\Sigma'(\texttt{o})}-\size{\Sigma(\texttt{o})}$ and $k$ are both constants independent of $\Sigma$.
    \item If $\s$ is an iteration statement,
    \[
    \srule{\cost{\Sigma}{\e}{i}{k_\e}\quad\Sigma_0:=\Sigma\enspace\cost{\Sigma_{j}[x\mapsto j]}{T_2(\s')}{\Sigma_{j+1}}{k_j},\text{for}\;\;0\le j< i.}{\cost{\Sigma}{\texttt{\textbf{for}(}x \texttt{<} \e \texttt{) } T_2(\s')}{\Sigma_{i}}{k_\e+\sum_{j=0}^{i-1}k_j}}{Loop}
    \]
    And $k = k_\e+\sum_{j=0}^{i-1}k_j + 5$, where the addend $5$ comes from the appended \texttt{o=o+o;} statement and the enclosing block.
    By applying the induction hypothesis, there exists a constant $d'$ such that
    \begin{equation}
        k_j\le d' \cdot (\size{\Sigma_{j+1}(\texttt{o})}-\size{\Sigma_{j}(\texttt{o})}),\text{ for }0\le j<i.
    \end{equation}
    We take $d := \max (k_\e + 5, d')$, which depends only on $s$. 
    Then it holds that
    $$k = k_\e + \sum_{j=0}^{i-1}k_j + 5 \le d \cdot (1 + \size{\Sigma_{i}(\texttt{o})}-\size{\Sigma_{0}(\texttt{o})}) = d \cdot (\size{\Sigma'(\texttt{o})}-\size{\Sigma(\texttt{o})}).$$
    The last equality holds as we add an extra \texttt{o=o+o;} right after the whole loop.
\end{itemize}
The other cases can be proved in the same manner.
\qed
\end{proof}

\subsection{Proof of Lemma~\ref{lem:expr-size}}\label{App_proof_of_lemma:expr-size}

\begin{proof}
Both proofs are by induction on the structure of $\evalTree{\Sigma}{\e}$.
\begin{itemize}
    \item Case \sref{Var}, $\e\equiv x$. Then $v=\Sigma(x)$ implies $\size{v}=\size{\Sigma(x)}\le\size{\Sigma}$.
    \item Case \sref{Const}, $\e\equiv \c$. Then $v=\smap{\c}{}$ implies
    $\size{v}\le\size{\Sigma}+k$, where $k=\size{\smap{\c}{}}$ is a constant.
        \item Case \sref{Op}, $\e\equiv \op(\widetilde{\e})$. There exists $\widetilde{v}$ such that $\semantics{\Sigma}{\widetilde{\e}}{\widetilde{v}}$, $\semantics{\Sigma}{\e_2}{v_2}$ and $v=\smap{\texttt{op}}{}(\widetilde{v})$. By the induction hypothesis, there exist constants $k_i, 1\le i\le m:=\arity(\op)$ such that $\size{v_i}\le\size{\Sigma}+k_i$. By Lemma \ref{lem:op-size}, there exists a constant $k_0$ such that
    \begin{equation}
        \size{v}=\size{\smap{\texttt{op}}{}(\widetilde{v})}\le\max_{1\le i\le m}\size{v_i}+k_0\le\size{\Sigma}+\max_{1\le i\le m} k_i+k_0.
    \end{equation}
    The proof is completed by setting $k=\max_{i=1}^{m}k_i+k_0$.
    \item Case \sref{Paren} is trivial.
\end{itemize}

If there exists a $\Gamma$ consistent with $\Sigma$ such that $\typing{\Gamma}{\f}{\e}{\typeIterInt}$ holds, then consider the following cases.
    \begin{itemize}
        \item Case~\sref{Var}, $\e\equiv x$.
        Since $\typing{\Gamma}{\f}{x}{\typeIterInt}$, we have $x\in\domi{\Gamma}$, which implies $\size{v}=\size{\Sigma(x)}\le\size{\restr{\Sigma}{\Gamma}}$.
    \item Case \sref{Const}, $\e\equiv \c$.
    Then $v=\smap{\c}{}$ implies
    $\size{v}\le\size{\restr{\Sigma}{\Gamma}}+\size{\smap{\c}{}}$.
    \item Case \sref{Op}, $\e\equiv \op(\widetilde{\e})$.
    For $1\le i\le m:=\arity(\op)$, there exists $v_i,t_i$ such that $\semantics{\Sigma}{\e_i}{v_i}$, $\typing{\Gamma}{\f}{\e_i}{t_i}$, and $v=\smap{\op}{}(\widetilde{v})$, $\tmap{\op}{}(\widetilde{t})\equiv\typeIterInt$.
    Therefore, it must be the case that $t_1=t_2\equiv\typeIterInt$.
    By the induction hypothesis, there exist constants $k_i$ such that $\size{v_i}\le\size{\restr{\Sigma}{\Gamma}}+k_i$.
    By Lemma \ref{lem:op-size}, there exists a constant $k_0$ rendering true the following equation.
    \begin{equation}
        \size{v}=\size{\smap{\op}{}(\widetilde{v})}\le\max_{1\le i\le m} \size{v_i}+k_0\le\size{\restr{\Sigma}{\Gamma}}+\max_{1\le i\le m} k_i+k_0.
    \end{equation}
    The proof is completed by setting $k=\max_{1\le i\le m}k_i+k_0$.
    \item Case \sref{Paren} is trivial.
    \end{itemize}
We are done.\qed
\end{proof}

\subsection{Proof of Lemma~\ref{lem:stmt-invariant}}\label{App_proof_of_lemma:stmt-invariant}

\begin{proof}
The proof is by induction.
We first focus on the case that $\l=\t$.
\begin{itemize}
    \item In case \tref{Asgmt}, $\typing{\Gamma}{\t}{\e}{\tt}$. By Lemma \ref{lem:expr-size} (1), there exists $v$ such that $\semantics{\Sigma}{\e}{v}$ and $\size{v}\le\size{\Sigma}+k$.
    Note that $\size{\Sigma'}=\size{\Sigma[x\mapsto v]} \le \max\{\size{\Sigma},\size{v}\} \le \size{\Sigma} + k$.
    We have $\size{\Sigma'}\dot{-}\size{\Sigma}\le k \le k \cdot \size{\restr{\Sigma}{\Gamma}}^0$.
    \item In case~\tref{Decl}, $\size{\Sigma'}\dot{-}\size{\Sigma}=\max\{\Sigma,\smap{t}{}\}-\size{\Sigma}=0$.
    \item In case \tref{Loop}, there exists $i$ such that $\semantics{\Sigma}{\e}{i}$.
    By the syntactic restriction, $\e=\texttt{\textbf{size}($\e'$)}$ and $\typing{\Gamma}{\t}{\e'}{\typeIterInt}$.
        By Lemma \ref{lem:expr-size} (2), there exists $v$ such that $\semantics{\Sigma}{\e'}{v}$ and $\size{v}\le\size{\restr{\Sigma}{\Gamma}}+k$, where $k$ is a constant. Therefore, $i=\smap{\opSize}{}(v)=\size{v}\le\size{\restr{\Sigma}{\Gamma}}+k$.
    \begin{itemize}
        \item $i\le0$.
        As $\Sigma'=\Sigma$, the result is trivial.
        \item $i>0$.
        By the induction hypothesis, we have that 
        \begin{equation}
            \size{\Sigma_{j+1}}\dot{-}\size{\Sigma_j}\le c \cdot \size{\restr{\Sigma_j[x\mapsto j]}{\Gamma[x\mapsto t]}}^{d},
        \end{equation}
        for $0\le j<i$, and
        \begin{equation}
            \size{\Sigma'}\dot{-}\size{\Sigma} \le c\cdot\sum_{j=0}^{i-1}\size{\restr{\Sigma_j[x\mapsto j]}{\Gamma[x\mapsto t]}}^{d}+\size{\Sigma_0[x\mapsto0]}-\size{\Sigma_0}.
        \end{equation}
        Since Proposition \ref{prop:invariant} shows that $\restr{\Sigma_j}{\Gamma[x\mapsto t]}\sim\restr{\Sigma_{j-1}[x\mapsto j-1]}{\Gamma[x\mapsto t]}$ for $1\le j<i$, one has that
        \begin{equation}
            \begin{aligned}
                \size{\restr{\Sigma_j[x\mapsto j]}{\Gamma[x\mapsto t]}} &= \max\{\size{j}, \size{\restr{\Sigma_j}{\Gamma[x\mapsto t]}}\} \\ &=\max\{\size{j},\size{\restr{\Sigma_{j-1}[x\mapsto j-1]}{\Gamma[x\mapsto t]}}\}\\
                &=\cdots\\
                &=\max\{\size{j},\size{\restr{\Sigma_0[x\mapsto0]}{\Gamma[x\mapsto t]}}\}\\
                &\le\max\{\size{i},\size{\restr{\Sigma}{\Gamma}}\}.
            \end{aligned}
        \end{equation}
        Then
        \begin{equation}
            \size{\Sigma'}\dot{-}\size{\Sigma} \le c \cdot (\size{\restr{\Sigma}{\Gamma}}+k)^{d+1}\le c \cdot k^{d+1} \cdot \size{\restr{\Sigma}{\Gamma}}^{d+1}.
        \end{equation}
    \end{itemize}
    \item In \tref{Seq}, we have $\size{\Sigma_i}\dot{-}\size{\Sigma_{i-1}}\le c_i \cdot \size{\restr{\Sigma_{i-1}}{\Gamma_{i-1}}}^{d_i}$ for $1\le i\le m$.
    Proposition \ref{prop:invariant} implies that $\restr{\Sigma_{i}}{\Gamma_i}=\restr{\Sigma}{\Gamma}$ for all $i$.
    Let $d:=\max_{1\le i\le m}d_i$ and $c := \max_{1 \le i \le m} c_i$.
    Then
    \begin{equation}
        \size{\Sigma'}\dot{-}\size{\Sigma} =\size{\Sigma_m}\dot{-}\size{\Sigma_{0}}\le c \cdot \sum_{i=1}^{m} \size{\restr{\Sigma_{i-1}}{\Gamma_{i-1}}}^{d}\le c \cdot m \cdot \size{\restr{\Sigma}{\Gamma}}^{d},
    \end{equation}
    since $m$ is a constant independent of $\Sigma$.
    \item Case~\tref{cond} and \tref{Block} are trivial.
\end{itemize}
If $\l=\f$, the proof resembles the above analysis, except for \sref{Seq}.
In this case, we can only obtain that $\size{\Sigma_i}\le\size{\Sigma_{i-1}}+O(\size{\Sigma_{i-1}}^{d_i^{k_i}})=O(\size{\Sigma_{i-1}}^{d_i^{k_i}})$.
Let $d$ denote $\max_{1\le i\le m}d_i$ and $k=\sum_{1\le i\le m}k_i$.
We finally have that $\size{\Sigma'}\dot{-}\size{\Sigma}=\size{\Sigma_m}\dot{-}\size{\Sigma_0}=O(\size{\Sigma}^{d^{k}})$.
\qed
\end{proof}

\subsection{Proof of Lemma~\ref{lemma:poly-con}}\label{App_proof_of_lemma:poly-con}

\begin{proof}
By construction.
Consider the following program:
\begin{lstlisting}[xleftmargin=3em,xrightmargin=3em]
int main(int x){
    iint z; z=x;
    int o;  o=1;
    for(i<size(z)){
        `$\cdots$`
        // Nested loops of depth `$\color{gray}d+1$`.
        for(j<size(z)){
            for(k<`$d$`){
                o=o+o;
            }
        }
    }
    return o/2;
}
\end{lstlisting}
There are $d+1$ nested loops in the program, where the outer $d$ loops have the upper bound of \texttt{\textbf{size}(z)}, and the innermost loop has the upper bound of $d$.
Through simple calculation, we know that $\smap{\p}{}(v)=2^{d\size{v}^d-1}$, and hence $\size{\smap{\p}{}(v)}=d\size{v}^d$.
It is trivial that $\size{p}=O(d\log d)$.
\end{proof}

\section{Missing proofs of Section \ref{sec:polyc}}\label{app-polyc}
\subsection{Proof of Lemma~\ref{lem:polyc-expr-size} and Lemma~\ref{lem:stmt-invariant} for PolyC}
\label{app:proof_polyc_size}
We prove these two lemmas in a mutually inductive manner.
\begin{proof}
By induction on the depth of function calls ${D}$.
More formally, we can define $D$ as the maximum number of the \sref{App} rule instances along any path from the root node to a leaf node in the corresponding evaluation derivation.

In the case of $D=0$, when expressions or statements do not contain any function calls, it aligns with the scenarios covered by the proofs of Lemma \ref{lem:expr-size} and Lemma \ref{lem:stmt-invariant}, respectively.

Assuming that the conclusions to be proven for any expression or statement with $D=i\geq 0$ have already been established, let us now consider the case of $D=i+1$.
Firstly, the conclusion holds for an expression $\e$ with $D$, which satisfies $\semantics{\Sigma}{\e}{v}$ and $\typing{\Gamma}{\f}{\e}{t}$.
Consider the structure of $\evalTree{\Sigma}{\e}$.
\begin{itemize}
    \item Case \sref{App}. In this case, assume $\Sigma(\texttt{fun})=(\widetilde{x},\widetilde{\s},\e')$. Then, we have $\semantics{\Sigma}{\widetilde{\e}}{\widetilde{v}}$, $\semantics{\emptyset[\widetilde{x}\mapsto\widetilde{v}]}{\widetilde{s}}{\Sigma'}$ and $\semantics{\Sigma'}{\e'}{v}$.
    Note that $\e$ itself includes one function call, so for $\widetilde{\e}$ or the $\widetilde{\s},\,\e'$ within the function body, the number of occurrences of \sref{App} along any path in their corresponding evaluation derivation is at least reduced by one.
    Therefore, they correspond to the case where $D\le i$. We can then apply the induction hypothesis to obtain:
    \begin{itemize}
        \item  If $\widetilde{v}=(v_1,\cdots,v_m)$, then for any $1\leq j\leq m$, $\size{v_j}\le\size{\Sigma}+O(\size{\restr{\Sigma}{\Gamma}}^{d_j})$ for some constant $d_j$. In particular, if $\e_j$ is an iterable expression, then $\size{v_j}=O(\size{\restr{\Sigma}{\Gamma}}^{d_j})$.
        Therefore, $\size{\emptyset[\widetilde{x}\mapsto\widetilde{v}]}\dot{-}\size{\Sigma}=O(\size{\restr{\Sigma}{\Gamma}}^{d'})$, where $d':=\max_{1\le j\le m} d_j$.
        \item By \tref{Fun}, there exists $\Gamma'$ such that $\typing{\emptyset[\;\widetilde{x}\mapsto\widetilde{t}\;]}{\t}{\widetilde{\s}}{\Gamma'}$ and $\typing{\Gamma'}{\t}{\e'}{t'}$.
        Applying the induction hypothesis again, we obtain
        \begin{equation}
        \size{\Sigma'}\dot{-}\size{\emptyset[\widetilde{x}\mapsto\widetilde{v}]}=O(\size{\restr{\emptyset[\widetilde{x}\mapsto\widetilde{v}]}{\emptyset[\widetilde{x}\mapsto\widetilde{t}]}}^{d''}),
        \end{equation}
        for some constant $d''$.
        \tref{App} ensures that $\widetilde{\e}$ and $\widetilde{t}$ are consistent, so $\size{\restr{\emptyset[\widetilde{x}\mapsto\widetilde{v}]}{\emptyset[\widetilde{x}\mapsto\widetilde{t}]}}=O(\size{\restr{\Sigma}{\Gamma}}^{d'})$, since only the values of iterable parameters remain after the restriction. Now, we obtain that
        \begin{equation}
            \size{\Sigma'}\dot{-}\size{\Sigma}=O(\size{\restr{\Sigma}{\Gamma}}^{\max\{d',d''\}}).
        \end{equation}
        \item In the end, we have $\size{v}\le\size{\Sigma'}+O(\size{\restr{\size{\Sigma'}}{\Gamma'}}^{d'''})$ for some $d'''$ by the induction hypothesis.
        By Proposition \ref{prop:invariant}, $\size{\restr{\Sigma'}{\Gamma'}}=\size{\restr{\emptyset[\widetilde{x}\mapsto\widetilde{v}]}{\emptyset[\widetilde{x}\mapsto\widetilde{t}]}}$.
        Hence, $\size{v}\le\size{\Sigma}+O(\size{\restr{\Sigma}{\Gamma}}^{d})$, where $d:=\max\{d',d'',d'''\}$ is a constant.
    \end{itemize}
    \item For all other cases, it suffices to observe that for any $\op\in\O_{+}$, if $\typing{\Gamma}{\f}{\op(\widetilde{x})}{t}$ and $\semantics{\Sigma}{\op(\widetilde{x})}{v}$, then $\size{v}\dot{-}\size{\Sigma}=O(\mathrm{poly}(\size{\restr{\Sigma}{\Gamma}}))$.
\end{itemize}

Concerning the special case where $\e$ is an iterable expression, it is sufficient to note that in the above proof of Case~\sref{App}, $\typing{\Gamma'}{\t}{\e'}{t'\equiv{\typeIterInt}}$ or \texttt{\textbf{istring}}.
Therefore, the third step should yield a stronger conclusion $\size{v}=O(\size{\restr{\size{\Sigma'}}{\Gamma'}}^{d})$.
Consequently, $\size{v}=O(\size{\restr{\Sigma}{\Gamma}}^{d})$.
The proofs for the other cases follow a similar pattern to those of Lemma \ref{lem:expr-size}.

Secondly, we further demonstrate that the corresponding conclusions hold for statements with $D=i+1$.
In case \tref{Asgmt}, the statement takes the form \texttt{$x$=$\e$;}, and $\typing{\Gamma}{\t}{\e}{t}$. For $\e$, it corresponds to the case of $D\leq i+1$. By applying the conclusions proved above, we have $\size{v}\dot{-}\size{\Sigma}=O(\size{\restr{\Sigma}{\Gamma}}^{d})$. Ultimately, $\size{\Sigma'}\dot{-}\size{\Sigma}\leq\size{v}\dot{-}\size{\Sigma}=O(\size{\restr{\Sigma}{\Gamma}}^{d})$.
The proofs for the remaining cases are similar to the proof of Lemma \ref{lem:stmt-invariant}.

Now, we have demonstrated that the conclusions hold for expressions and statements with any depth of function calls $D$ that satisfy the specified conditions. This completes the proofs of Lemma \ref{lem:polyc-expr-size} and Lemma \ref{lem:stmt-invariant} (PolyC version).
\qed
\end{proof}

\subsection{Proof of Proposition~\ref{prop:poly-encode}}
\begin{proof}
Given a program $\p$ without input, the syntax parsing step can be finished within polynomial time of $\size{\p}$.
However, as shown in Lemma~\ref{lem:stmt-invariant}, the instruction count of $\p$ is bound by $O(\size{\Sigma}^{d^{k}})$, where both $d$ and $k$ are no greater than the size of the program.
Consequently, it is enough to solve $\I_0(\p)$ in $2^{O(\size{p}^{\size{\p}})}$-time using a Turing machine.

For any $L \in 2$-$\mathbf{EXP}$, assume that there is a Turing machine $M$ deciding $L$ in $2^{2^{\mathrm{poly}(n)}}$-time, where $n$ is the size of the input $x$.
Modify $\p$ in the proof of Lemma~\ref{lemma:poly-con} by setting $d=2$ and repeating the loop part $n^{k}$ times. It is clear that $\size{\smap{p}{}(2)}=\Omega(2^{2^{n^{k}}})$.
This shows that we can construct a variable of double-exponential length as a clock.
Then, using the method outlined in Lemma~\ref{lem:simulate}, we can construct a $\mathrm{poly}(n)$-size program $\p'$ to simulate the execution of $M$ and encode the input $x$ as a constant in $\p'$.
Thus, we have constructed a polynomial-time reduction from $L$ to $\I_0$ with $x\mapsto\p'$.
We are done.
\qed
\end{proof}

\subsection{Proof of Proposition~\ref{prop:unfold}}
\label{app:proof_unfold}
\begin{proof}
We say that a sequence of statements $\widetilde{\s}$ is simple if it consists of two parts: i) some declaration statements, and ii) a loop \texttt{\textbf{for}(i<\textbf{size}(y)) $\s'$}, where $\s'$ does not contain any loops or function calls.
We now transform the statements first.

\begin{claim}
    Let $\widetilde{\s}$ be a sequence of statements that does not contain any function calls, and \texttt{y} be a fresh variable
    There exists a bound $t_0(n)=n^{\size{\widetilde{\s}}^{O(\size{\widetilde{\s}})}}$ and a $O(\mathrm{poly}(\size{\widetilde{\s}}))$-size simple statement $\widetilde{\s}_0$ such that
    \begin{enumerate}
        \item the loop bound in $\widetilde{\s}_0$ is \texttt{\textbf{size}(y)}, and
        \item for any store $\Sigma$ such that $\semantics{\Sigma}{\widetilde{\s}}{\Sigma'}$ for some $\Sigma'$, and $|t|\ge t_0(\size{\Sigma})$, one has that $\semantics{\Sigma[\texttt{y}\mapsto t]}{\widetilde{\s}_0}{\Sigma'[\texttt{y}\mapsto t]}$.
    \end{enumerate}
    The notation $\size{\widetilde{\s}}:=\sum_{i=1}^m\size{\s_m}$, i.e., the sum of the bit-size of all statement in $\widetilde{\s}$.
\end{claim}

\begin{proof}
    Assume that $\widetilde{\s}=(\s_1,\dots,\s_m)$.
    We prove by induction on $m\ge 1$.
    If $m=1$, we prove by induction on the structure of statements.
    \begin{enumerate}
        \item Case \texttt{x=$\e$;}.
        Let $\widetilde{\s}_0:=\texttt{\textbf{for}(i<\textbf{size}(y)) \textbf{if}(i<1) x=$\e$; \textbf{else} \textbf{break};} $ and $t_0:=1$.
        Since $\e$ does not invoke any function calls, $\widetilde{\s}_0$ and $t_0$ satisfy all the conditions.
        \item Case $\texttt{\textbf{for}(j<\textbf{size}(z)) $\s'$}$.
        Applying the induction hypothesis, we can obtain $\widetilde{\s}_0'$ and $t_0'$.
        Assume that \begin{equation}
        \widetilde{\s}_{0}'\equiv \widetilde{\s}_{1}\texttt{ \textbf{for}(i<\textbf{size}(y)) }\s_{2},
    \end{equation}
    where $\widetilde{\s}_{1}$ corresponds to the declaration stage.  
    We can always rename variables in a way that \texttt{i} and \texttt{j} are two fresh variables.
    Let $\widetilde{\s}_0$ be the following statements.
\begin{lstlisting}
    `$\widetilde{\s}_1$`
    int i, j;
    for(k<size(y)) {
        if(j<size(z){
            `$\s_2$`
            i+=1;
            if(i>=size(y)/size(z)) {i=0; j+=1;}
        } else break;
    }
\end{lstlisting}
        According to Lemma~\ref{lem:stmt-invariant}, $\widetilde{\s}$ can only produce value bounded by a polynomial $p=O(n^{\size{\widetilde{\s}}^{\size{\widetilde{\s}}}})$.
        Let $t_0(n):=p(n)t'_0(p(n))$ and we are done.
        \item Other cases are trivial.
    \end{enumerate}

    If $\widetilde{\s}=(\widetilde{\s}_{\le m},\s_{m+1})$, then we first apply the induction hypothesis to obtain the corresponding simple statements and functions $\widetilde{\s}_{\le m, 0}$, $t_{\le m}$, $\widetilde{s}_{m+1,0}$ and $t_{m+1}$.
    Assume that
    \begin{equation}
        \widetilde{\s}_{\le m, 0}\equiv\widetilde{\s}_{\le m, 1}\texttt{ \textbf{for}(i<\textbf{size}(y)) }\s_{\le m, 2},
    \end{equation}
    \begin{equation}
        \widetilde{\s}_{m+1, 0}\equiv\widetilde{\s}_{m+1, 1}\texttt{ \textbf{for}(i<\textbf{size}(y)) }\s_{m+1, 2},
    \end{equation}
    where $\widetilde{\s}_{\cdot,1}$ corresponds to the declaration stage.
    Let $\widetilde{\s}_0$ be the following sequence.
    \begin{equation}
        \widetilde{\s}_{\le m, 1}\texttt{; }\widetilde{\s}_{m+1, 1}\texttt{ \textbf{for}(i<\textbf{size}(y)) \textbf{if}(i<\textbf{size}(y)/2) $\s_{\le m,2}$ \textbf{else} $\s_{m+1,2}$}.
    \end{equation}
    According to Lemma~\ref{lem:stmt-invariant}, $\widetilde{\s}_{\le m}$ can only produce value bounded by a polynomial $p=O(n^{\size{\widetilde{\s}_{\le m}}^{\size{\widetilde{\s}_{\le m}}}})$.
    Let $t_0:=2\max(t_{\le m}, t_{m+1} \circ \p)$.
    It is easy to see that $\widetilde{s}_0$ and $t_0$ satisfy the three conditions.
    \qed
\end{proof}

Then we proceed to prove the origin proposition.
Without loss of generality, we assume that $\p$ does not invoke any function calls since all the function calls can be rewritten in an inline style.
Assume that the body of the main function consists of statements $\widetilde{\s}$ and \texttt{\textbf{return} $\e$}.
By our claim, we can obtain a polynomial-size simple statement $\s_0$ and a bound $t_0$ fulfilling the above conditions.
Let $\p_0:=\tt\;\texttt{main}(\tt_1\;x_1,\dots,\tt_m\;x_m,\texttt{iint}\;y)\texttt{\{$\widetilde{\s}_0$ \textbf{return} $\e$\}}$, where $\tt,\tt_1,\dots,\tt_m$ is the same as $\p$.
Clearly, $\p_0$ is simple.
It suffices to verify that $\p_0$ satisfies the second condition.
For any inputs $\widetilde{v}\in \V^m$ and $|t|\ge t_0(\size{\widetilde{v}})$, 
there exists some $\Sigma'$ such that $\semantics{\emptyset[\widetilde{x}\mapsto\widetilde{v}]}{\widetilde{s}}{\Sigma'}$, and $\semantics{\Sigma'}{\e}{\tmap{\p}{}(\widetilde{v})}$.
Since $|t|\ge t_0(\size{\widetilde{v}})=t_0(\emptyset[\widetilde{x}\mapsto\widetilde{v}])$, one has that $\semantics{\emptyset[\widetilde{x}\mapsto\widetilde{v}][\texttt{y}\mapsto t]}{\widetilde{\s}_0}{\Sigma'[\texttt{y}\mapsto t]}$.
Note that \texttt{y} is fresh.
Hence, $\semantics{\Sigma'[\texttt{y}\mapsto t]}{\e}{\tmap{\p}{}(\widetilde{v})}$, which implies $\smap{\p_0}{}(\widetilde{v},t)=\smap{\p}{}(\widetilde{v})$.

\qed
\end{proof}

\end{document}